\documentclass[journal, final, onecolumn, 12pt, letterpaper]{IEEEtran}
\usepackage[utf8]{inputenc}
\usepackage{amsmath}
\usepackage{amsfonts}
\usepackage{amssymb}
\usepackage{commath}
\usepackage{amsthm}
\usepackage{graphicx}
\usepackage{bm}
\usepackage{multirow}
\usepackage{epstopdf}
\usepackage{physics}
\usepackage{mathtools}
\usepackage{cases}

\usepackage{caption}
\newtheorem{theorem}{Theorem}
\newtheorem{lemma}[theorem]{Lemma}

\ifCLASSOPTIONcompsoc
\usepackage[caption=false,font=normalsize,labelfont=sf,textfont=sf]{subfig}
\else
\usepackage[caption=false,font=footnotesize]{subfig}
\fi

\newcommand\scalemath[2]{\scalebox{#1}{\mbox{\ensuremath{\displaystyle #2}}}}

\makeatletter
\renewcommand*\env@matrix[1][\arraystretch]{%
  \edef\arraystretch{#1}%
  \hskip -\arraycolsep
  \let\@ifnextchar\new@ifnextchar
  \array{*\c@MaxMatrixCols c}}
\makeatother

\title{Optimized Signaling of Binary Correlated Sources over GMACs}
\renewcommand{\arraystretch}{1.5}
\DeclareMathOperator*{\argmax}{arg\, max}
\DeclareMathOperator*{\argmin}{arg\, min}

\author{Jian-Jia~Weng,~Fady~Alajaji, and~Tam{\'a}s~Linder%
\thanks{The authors are with the Department of Mathematics and Statistics, Queen’s University, Kingston, ON K7L 3N6, Canada (email: jian-jia.weng@queensu.ca, fady@mast.queensu.ca, linder@mast.queensu.ca).}
\thanks{This work was supported in part by NSERC of Canada.}}

\begin{document}

\maketitle
\begin{abstract}
This work focuses on the construction of optimized binary signaling schemes for two-sender uncoded transmission of correlated sources over non-orthogonal Gaussian multiple access channels. 
Specifically, signal constellations with binary pulse-amplitude-modulation are designed for two senders to optimize the overall system performance. 
Although the two senders transmit their own messages independently, it is observed that the correlation between message sources can be exploited to mitigate the interference present in the non-orthogonal multiple access channel. 
Based on a performance analysis under joint maximum-a-posteriori decoding, optimized constellations for various basic waveform correlations between the senders are derived. 
Numerical results further confirm the effectiveness of the proposed design. 
\end{abstract}


\section{Introduction}
The requirement of transmitting correlated information appears in many practical scenarios. For example, nearby measurement stations regularly report observed temperatures to a control center to track environmental change. For transmitting correlated sources over Gaussian multiple access channels (GMACs), the pioneering study in \cite{CT80} proposed a random coding scheme to establish reliable communication from the channel capacity perspective. 
Using powerful channel codes, e.g., low-density parity-check codes \cite{Gallager63} or turbo codes \cite{Turbo}, practical code constructions with capacity approaching performance were also given for various GMACs \cite{AA09}-\cite{AA14}. Unfortunately, most of these codes incur relatively high computational complexity and long decoding delay. 
An alternative approach to channel coding is uncoded transmission in which each source symbol is directly mapped to one channel input signal. This simple scheme is particularly suitable for resource-limited systems such as wireless sensor networks \cite{Liu10}. However, in the absence of the protection provided by channel codes, recovering the transmitted data from the received noisy signal becomes challenging. 

In this paper, we study the optimization of uncoded transmission of correlated sources over GMACs. Our objective is to design binary signaling schemes for each sender such that the system joint error rate is minimized. The basic setup is briefly summarized as follows. 
The two senders are assumed to employ binary pulse-amplitude-modulation (BPAM) such that each sender has its own energy constraint. 
The GMAC we consider is a non-orthogonal channel. The two BPAM signals are transmitted in the same time slot and frequency band and hence multiple access interference will occur if the senders' basic pulse waveforms are not orthogonal. Furthermore, the two senders are assumed to transmit their own messages independently. Cooperative transmission is excluded in this paper because it is usually infeasible for resource-limited networks with separated transmitters. Lastly, a joint maximum-a-posteriori (MAP) decoder which can exploit the correlation between the source messages is used at the receiver. We note that a similar problem has been tackled recently in \cite{Tyson15}, in which an orthogonal GMAC was considered. Our work can be viewed as a substantial generalization of \cite{Tyson15}. 

Under the above setting, it can readily be seen that using an identical BPAM signaling scheme at both senders is inadequate because it results
in a combined constellation for the transmitted pair of messages in which the constellation of one user is exactly superposed to the constellation
of the other user. In this case, the receiver cannot decode the received signal without any error, even when the transmission is noise-free. To resolve this ambiguity, \cite{JH11} and \cite{JH13} respectively propose a rotation scheme and an energy allocation scheme. 
While these schemes aim to enlarge the constrained constellation capacity for the transmission of independent and uniformly distributed sources over GMACs, the proposed ideas may also improve the error rate performance for correlated and non-uniform sources. 
However, as reported in \cite{Korn03}-\cite{LW12}, symmetric constellations are often not optimal for non-uniformly distributed sources. Using the modulated signals obtained by either the rotation or energy allocation scheme is then likely to be sub-optimal. Instead of significantly altering the conventional antipodal BPAM constellation, we propose to directly design signals. 
In this approach, we explicitly construct constellations by analytically optimizing the system's exact error rate or its upper bound for high signal-to-noise ratios (SNRs). More importantly, the correlation between sources is not only exploited to boost the decoding performance, but it is also used to mitigate the interference between the two independently transmitted signals. 

We next briefly review some prior work related to the subject of this paper.  
In \cite{Cheng89}, the authors characterize the capacity region of the two-sender GMAC using PAM signals based on the notions of root-mean-square and factional out-of-band energy. 
Achievable rates for the two-sender GMAC with uncoded PAM signals are derived in \cite{DT09}. 
Prior work on designing non-binary constellations for non-uniformly distributed sources sent over point-to-point channels, e.g., \cite{Nguen05}-\cite{LW12}, \cite{Moore10}, can be helpful when higher order modulation schemes are considered for the GMAC. 
When the number of senders or the modulation order increases, successive interference cancellation decoding \cite{Verdu98} can be employed to reduce the computational complexity of the receiver. 
Note that although our transmission model is simple, the results obtained in this paper can be potentially applied to wireless ad-hoc networks \cite{Tonguz06}, cognitive radios \cite{Tarokh05}, and non-orthogonal multiple access in the fifth generation (5G) mobile systems \cite{Dai15}.  

The rest of this paper is organized as follows. In Section~II, we describe the system with a two-sender GMAC and analyze its error rate performance under joint MAP decoding. In Section~III, the design procedure for correlated pulse waveforms is presented, and explicit optimized constellations are derived. In Section~IV, the performance of the proposed signaling schemes is systematically assessed via simulation. Conclusions and future works are drawn in Section~V.

\section{System Description and Error Rate Analysis}
\subsection{GMAC Transmission System}
The transmission system we study is depicted in Fig.\ \ref{BlockDig}. 
The system comprises two senders and one receiver. 
In each time slot, the senders simultaneously transmit their binary source messages over a multiple access channel with additive white Gaussian (AWGN) noise. 
The source messages are assumed to be correlated, and hence a joint MAP decoder is employed at the receiver to minimize the joint symbol error rate. 
The binary messages of sender $1$ and sender $2$ are denoted by $U$ and $V$, respectively, and they have the joint probability distribution $p_{uv} \triangleq \Pr\left(U=u, V=v\right)$, for $u, v\in\{0,1\}$. Let $p_1=\Pr(U=0)$ and $p_2=\Pr(V=0)$. 
The joint distribution can be also described in terms of $p_1$, $p_2$, and the sources' correlation coefficient as follows 
\begin{eqnarray}
p_{00} &=& 1 - (1 - p_1) - (1 - p_2) + p_{11}\nonumber\\
p_{10} &=& (1 - p_1) - p_{11}\nonumber\\
p_{01} &=& (1 - p_2) - p_{11},\nonumber
\end{eqnarray}
and 
\[
p_{11} = \gamma_{\text{m}}\sqrt{p_1(1 - p_1)p_2(1 - p_2)} + (1 - p_1)(1 - p_2). 
\]
where $\gamma_{\text{m}}=\text{Cov}(U, V)/(\sigma_{U}\sigma_{V})$ denotes the correlation coefficient between $U$ and $V$, $\text{Cov}(\cdot, \cdot)$ is the covariance, and $\sigma_U$ and $\sigma_V$ are the standard deviations of $U$ and $V$, respectively. 
To avoid uninteresting cases, we will assume $p_{uv}>0$ for all $u$ and $v$. 

\begin{figure}[!tb]
\centering
\includegraphics[draft=false, scale=0.5]{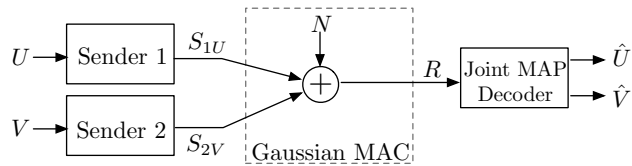} 
\caption{The block diagram of the GMAC transmission system.}
\label{BlockDig}
\end{figure}

To transmit data over the GMAC system, sender $j$ uses BPAM with waveform signal $a_{jb}\phi_j(t)$, where $a_{jb}$ denotes the amplitude for modulating binary source message $b\in\{0, 1\}$ and $\phi_j(t)$ is the sender's basic BPAM pulse waveform. 
The $\phi_j(t)$'s are assumed to be of finite duration $T$ and unit energy, i.e., $\int_{0}^{T}\phi_j^2(t)\dif t=1$. 
The correlation between $\phi_1(t)$ and $\phi_2(t)$ is denoted by $\gamma_{\phi}$, where $\gamma_{\phi}=\int_{0}^{T}\phi_1(t)\phi_2(t)\dif t$ and $-1\le\gamma_{\phi}\le 1$. 
Instead of using the continuous time description of signals, one can describe the waveform signals in a two-dimensional signal space by choosing $\phi_1(t)$ and $(\phi_2(t)-\gamma_{\phi}\phi_1(t))/(\int_{0}^{T}\phi_2(t)-\gamma_{\phi}\phi_1(t)\dif t)$ to form an orthonormal basis. For simplicity, the two basis vectors are identified with the real and imaginary axes on the complex plane. The waveform signal $a_{jb}\phi_j(t)$ can be now equivalently described by a point $S_{jb}$ on the complex plane obtained by projecting the waveform signal onto the signal space. 

By this choice of basis, the signal points $S_{10}$ and $S_{11}$ are located on the real axis of the complex plane with values $S_{10}=a_{10}$ and $S_{11}=a_{11}$, while the points $S_{20}$ and $S_{21}$ are generally complex-valued. 
There are two cases where both $S_{20}$ and $S_{21}$ are either purely real-valued or imaginary-valued. When $\gamma_{\phi}=0$, $S_{20}$ and $S_{21}$ lie on the imaginary axis of the complex plane with values $S_{20}=ia_{20}$ and $S_{21}=ia_{21}$, where $i$ denotes the imaginary unit. 
In this case, the basic pulse waveforms are orthogonal and thus no interference from the other sender will be introduced during transmission, and the transmitted signals are only perturbed by the channel noise. 
Another special case is when $\gamma_{\phi}=\pm 1$, i.e., $\phi_2(t)=\pm \phi_1(t)$. In this case, $S_{20}$ and $S_{21}$ are on the real axis. 
Here, strong multiple access interference occurs due to the high correlation between the transmitted signals of the two senders. 
A simple example is when both senders employ the same BPAM scheme, thereby producing $S_{10}=S_{20}$ and $S_{11}=S_{21}$.
Later, we will see that even under strong multiple access interference, it is possible to design appropriate BPAM signals to mitigate the interference and improve the quality of transmission.
 
Let $\mathcal{S}_{j}$ represent the BPAM signal constellation of sender~$j$ so that $\mathcal{S}_1=\{S_{10}, S_{11}\}$ and  $\mathcal{S}_2=\{S_{20}, S_{21}\}$. 
We additionally impose an average energy constraint on the constellation given by 
\begin{equation}
p_j|S_{j0}|^2+(1-p_j)|S_{j1}|^2= E_j,\ j=1, 2,
\label{PowerConst}
\end{equation}
where $|\cdot|$ denotes absolute value (magnitude) and $E_j$ is the average energy for transmitting sender $j$'s input source message. 
When the source messages $(U, V)$ are sent over the GMAC, the received signal at the output of the matched filter can be written as
\begin{equation}
R = S_{1U} + S_{2V} + N, 
\label{R}
\end{equation}
where $N$ denotes complex-valued zero-mean Gaussian noise with variance $\sigma^2$ per dimension, having independent components that are also independent of $(U, V)$. 
Letting $A_{UV}=S_{1U}+S_{2V}$, (\ref{R}) can be written as $R=A_{UV}+N$. 
Let $\mathcal{A}=\{S_{1u}+S_{2v}: u, v\in\{0,1\}\}$ denote the constellation that contains all combined signal points $A_{UV}$. 
Examples of the constellations $\mathcal{S}_1$, $\mathcal{S}_2$, and $\mathcal{A}$ are shown in Fig.\ \ref{ConstellationEx}. 
Note that to avoid harmful interference, it is necessary to use constellations such that the mapping from $\mathcal{S}_1\times\mathcal{S}_2$ to $\mathcal{A}$ given by $A_{UV}=S_{1U}+S_{2V}$ is bijective.

\begin{figure*}[!t]
\centering
\subfloat[$\gamma_{\phi}=\pm 1$]{\includegraphics[draft=false, scale=0.43]{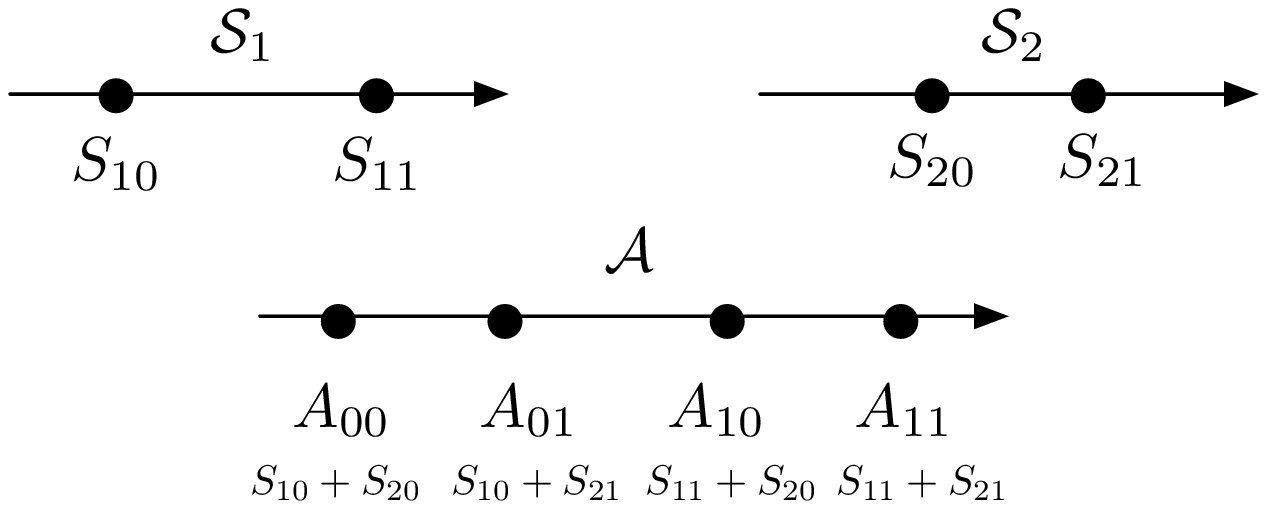}
\label{1DFig}}
\qquad
\subfloat[$\gamma_{\phi}\neq\pm 1$]{\includegraphics[draft=false, scale=0.43]{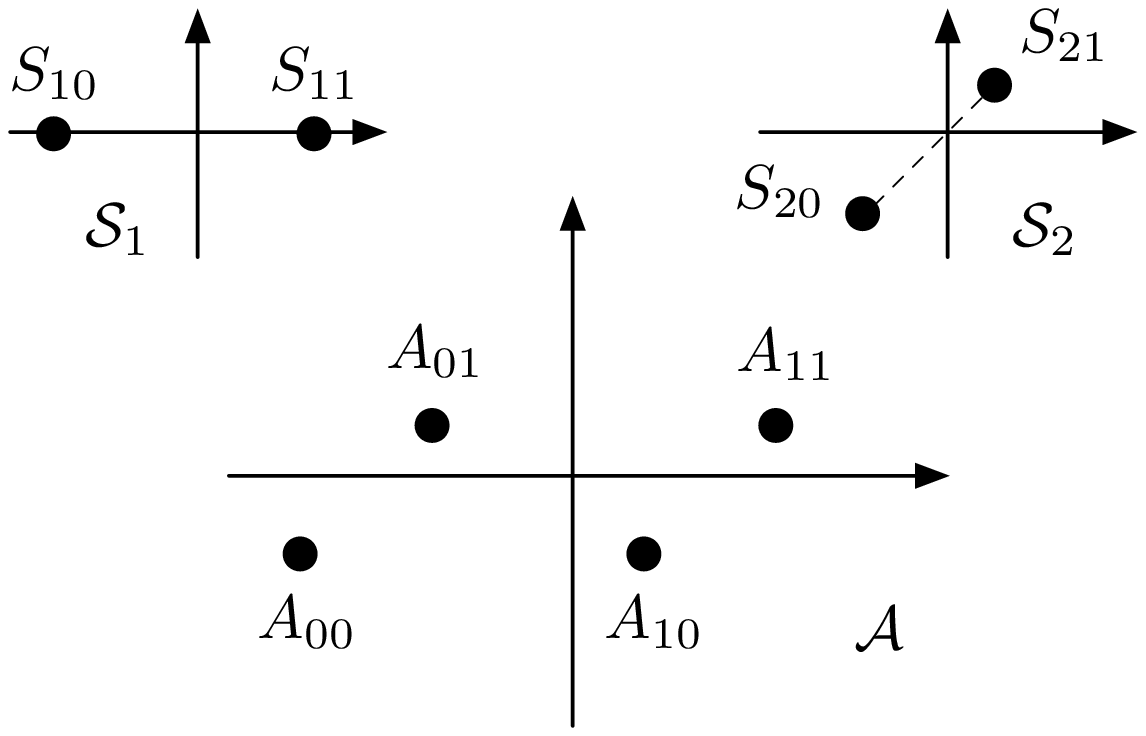}
\label{2DFig}}
\caption{An illustration of signal constellations of the two senders and the combined constellation.}
\label{ConstellationEx}
\end{figure*}

Suppose that $(U, V)=(u, v)$ is sent. 
When the receiver observes signal $R=r$, the (optimal) joint MAP decoder generates an estimate of the transmitted source messages based on the joint source distribution $p_{UV}$ and the observation $r$ via the decision rule 
\begin{IEEEeqnarray}{lCl}
(\hat{u}, \hat{v}) &=& \argmax_{(l,m)\in\{0,1\}^2}\Pr(U=l, V=m|R=r)\nonumber\\
&=&\argmax_{(l,m)\in\{0,1\}^2}p_{lm}\cdot\frac{1}{2\pi\sigma^2}\exp\left(\frac{-|r - A_{lm})|^2}{2\sigma^2}\right)\nonumber\\
&=&\argmax_{(l,m)\in\{0,1\}^2} \ln p_{lm} + \frac{2\Re[rA_{lm}^{*}]-|A_{lm}|^2}{2\sigma^2},\nonumber
\label{MAPRule}
\end{IEEEeqnarray}
where $\Re[z]$ and $z^{*}$ denote the real part and the conjugate of the complex number $z$, respectively. 
For convenience, we define the random variable $H_{lm}\triangleq\ln p_{lm} + (2\Re[RA_{lm}^{*}]-|A_{lm}|^2)/2\sigma^2$. 
Given $R=r$, the realization of $H_{lm}$, denoted by $h_{lm}$, can be viewed as a decision score for $A_{lm}$. The joint MAP decoder simply outputs the pair of source messages with the highest score so as to minimize the probability of erroneous detection. 

\subsection{The Error Rate Performance of the Joint MAP Decoder}
The conditional probability $P_{\text{c}, uv}$ of correct decoding given that $(U, V)=(u, v)$ is sent over the channel is given by
\begin{IEEEeqnarray*}{l}
P_{\text{c}, uv}=\Pr(\bigcap_{(l,m)\in\{0, 1\}^2: (l, m)\neq(u,v)}\{H_{uv}-H_{lm}> 0\}).
\label{Puv}
\end{IEEEeqnarray*}
To simplify the notation, we define the scaled difference metric  
\begin{IEEEeqnarray}{lCl}
\Delta_{uv, lm}&\triangleq&-\sigma^2\cdot(H_{uv}-H_{lm})\nonumber\\
&=&\Re[R(A_{lm}-A_{uv})^{*}]-\frac{|A_{lm}|^2-|A_{uv}|^2}{2} - \sigma^2\ln\frac{p_{uv}}{p_{lm}}.\label{SDM}
\end{IEEEeqnarray}
By substituting $R=A_{uv}+N$ in (\ref{SDM}), we obtain
\begin{IEEEeqnarray}{lCl}
\Delta_{uv, lm}&=&\Re[N(A_{lm}-A_{uv})^{*}] - \frac{|A_{lm}-A_{uv}|^2}{2} - \sigma^2\ln\frac{p_{uv}}{p_{lm}}. \label{Duvlm}
\end{IEEEeqnarray}
Note that $\Delta_{uv, lm}$ is a Gaussian random variable whose mean and variance are respectively given by
\begin{IEEEeqnarray}{c}
\mu_{uv, lm} \triangleq - \frac{|A_{lm}-A_{uv}|^2}{2} - \sigma^2\ln\frac{p_{uv}}{p_{lm}}
\label{Duvlmmu}
\end{IEEEeqnarray}
and
\begin{IEEEeqnarray}{c}
\sigma^2_{uv, lm}\triangleq \sigma^2\cdot|A_{lm}-A_{uv}|^2.
\label{Duvlmsigma}
\end{IEEEeqnarray} 
The system's error rate $P_{\text{err}}$ under joint MAP decoding can be written as
\begin{IEEEeqnarray}{lCl}
P_{\text{err}}&=&\sum_{(u, v)}p_{uv}\Pr\left((u,v)\neq\argmin_{(l,m)}\Delta_{uv, lm}\right)\label{PEP1}\nonumber\\
&=&1 - \sum_{(u, v)}p_{uv}\Pr\left((u,v)=\argmin_{(l,m)}\Delta_{uv, lm}\right)\nonumber\\
&=&1 - \sum_{(u, v)}p_{uv}\Pr\left(\Delta_{uv, lm}<0\ \text{for all}\ (l, m)\neq (u, v)\right)\nonumber\\
&=&1 - \sum_{(u, v)}p_{uv}P_{\text{c}, uv}. 
\label{MAPPb0}
\end{IEEEeqnarray}

In general, the probabilities $P_{\text{c}, uv}$ in (\ref{MAPPb0}) cannot be easily determined because the decision regions of the combined signal points on the complex plane are often of irregular shape, requiring complicated numerical integration. 
Nevertheless, when $\gamma_{\phi}=\pm 1$, the decision regions become intervals on the real line. For this case, we can directly identify the decision regions and calculate the probabilities $P_{\text{c}, uv}$. 
Specifically, for $\gamma_{\phi}=\pm 1$, the inequality $\Delta_{uv, lm}< 0$ can be expressed as
\begin{IEEEeqnarray}{c}
D_{uv, lm}\cdot \text{Re}[N]<  \frac{D_{uv, lm}^2}{2} + \sigma^2\ln\frac{p_{uv}}{p_{lm}}, 
\label{deltaeq} 
\end{IEEEeqnarray}
where $D_{uv, lm}\triangleq A_{uv}-A_{lm}=(S_{1u}+S_{2v})-(S_{1l}+S_{2m})$. 
Given $D_{10, 00}$, $D_{01, 00}$, and $|D_{10, 00}|-|D_{01, 00}|$, (\ref{deltaeq}) specifies a range of values in the form of an interval that $\text{Re}[N]$ can take. 
For each pair $(u, v)$, the decision region is specified by three such inequalities. 
Letting $\Omega$ be the intersection of the corresponding intervals, we have
\begin{IEEEeqnarray}{c}
P_{\text{c}, uv}=\int_{\Omega} \left(\frac{1}{\sqrt{2\pi\sigma^2}}\right)\exp\left(\frac{-t^2}{2\sigma^2}\right)\dif t
\label{Pb1D}
\end{IEEEeqnarray}
and the overall error probability is immediately obtained via (\ref{MAPPb0}). 
A detailed example for this procedure is given in the Appendix.

For $\gamma_{\phi}\neq\pm 1$, we can combine the techniques developed in \cite{Tyson15} and \cite{Leszek06} to find the error probability. 
First, based on (\ref{Duvlm}), it can be verified that
\begin{IEEEeqnarray}{c}
\Delta_{uv, \bar{u}\bar{v}}=\Delta_{uv, \bar{u}v} + \Delta_{uv, u\bar{v}} + \alpha_{uv}
\label{DuvlmIden}
\end{IEEEeqnarray}
with $\bar{e}=0$ if $e=1$ and $\bar{e}=1$ if $e=0$, and
\begin{equation*}
\alpha_{uv} = \left\{
\begin{array}{ll}
\sigma^2\ln\frac{p_{00}p_{11}}{p_{01}p_{10}} - \zeta,\ \text{if}\ (u, v)\in\{(0, 0), (1, 1)\}\\
\sigma^2\ln\frac{p_{01}p_{10}}{p_{00}p_{11}} + \zeta,\ \text{if}\ (u, v)\in\{(0, 1), (1, 0)\},\\
\end{array}\right.
\label{Alpha}
\end{equation*}
where $\zeta\triangleq\Re[(A_{10}-A_{00})(A_{01}-A_{00})^{*}]$. 
Using (\ref{DuvlmIden}), we have 
\begin{IEEEeqnarray*}{l}
P_{\text{c}, uv}=\Pr\Big(\{\Delta_{uv, \bar{u}v}<0\} \cap \{\Delta_{uv, u\bar{v}}<0\}\cap \{\Delta_{uv, \bar{u}v}+\Delta_{uv, u\bar{v}}+\alpha_{uv}<0\}\Big).\nonumber
\label{2DProb}
\end{IEEEeqnarray*} 
Since each $\Delta_{uv, lm}$ is an affine function of $N$, $\Delta_{uv, \bar{u}v}$ and $\Delta_{uv, u\bar{v}}$ are jointly Gaussian with joint probability density function (pdf)
\begin{IEEEeqnarray}{l}
f_{\Delta_{uv, \bar{u}v}, \Delta_{uv, u\bar{v}}}(x, y)=\frac{1}{2\pi\sigma_{uv, \bar{u}v}\sigma_{uv, u\bar{v}}\sqrt{1-\gamma^2}}\exp\left(\frac{-1}{2(1-\gamma^2)}\left[\left(\frac{x-\mu_{uv, \bar{u}v}}{\sigma_{uv, \bar{u}v}}\right)^2\right.\right.\nonumber\\
\left.\left.\qquad\qquad\qquad\qquad\qquad\qquad\qquad-2\gamma\frac{(x-\mu_{uv, \bar{u}v})(y-\mu_{uv, u\bar{v}})}{\sigma_{uv, \bar{u}v}\sigma_{uv, u\bar{v}}}\right.\right.
\left.\left.+\left(\frac{y-\mu_{uv, u\bar{v}}}{\sigma_{uv, u\bar{v}}}\right)^2\right]\right),\nonumber
\end{IEEEeqnarray}
where $\mu_{uv, lm}$ and $\sigma_{uv, lm}$ are respectively given in (\ref{Duvlmmu}) and (\ref{Duvlmsigma}), and $\gamma$ denotes the correlation coefficient between $\Delta_{uv, \bar{u}v}$ and $\Delta_{uv, u\bar{v}}$ given by
\begin{IEEEeqnarray}{l}
\gamma \triangleq\frac{\text{Cov}(\Delta_{uv, \bar{u}v}, \Delta_{uv, u\bar{v}})}{\sigma_{uv, \bar{u}v}\sigma_{uv, u\bar{v}}}=\frac{\zeta}{\sigma^2|A_{\bar{u}v}-A_{uv}||A_{u\bar{v}}-A_{uv}|}.\nonumber\\*
\label{gamma}
\end{IEEEeqnarray}
Let $\Lambda_1\triangleq (\Delta_{uv, \bar{u}v}-\mu_{uv, \bar{u}v})/\sigma_{uv, \bar{u}v}$ and $\Lambda_2\triangleq (\Delta_{uv, u\bar{v}}-\mu_{uv, u\bar{v}})/\sigma_{uv, u\bar{v}}$. Then, 
\begin{IEEEeqnarray}{rCl}
P_{\text{c}, uv}&=&\int_{-\infty}^{\frac{-\mu_{uv, u\bar{v}}}{\sigma_{uv, u\bar{v}}}}\int_{-\infty}^{\frac{-\mu_{uv, \bar{u}v}}{\sigma_{uv, \bar{u}v}}} \frac{1}{2\pi\sqrt{1-\gamma^2}}\exp\left(\frac{-1}{2(1-\gamma^2)}\left[\lambda_1^2-2\gamma\lambda_1\lambda_2+\lambda_2^2\right]\right)\dif \lambda_1\dif \lambda_2-\beta
\label{2DER}
\end{IEEEeqnarray}
with $\beta=0$ for $\alpha_{uv}\le 0$ and
\[
\beta = \int_{-\alpha_{uv}}^{0}\int_{-y-\alpha_{uv}}^{0}f_{\Delta_{uv, \bar{u}v}, \Delta_{uv, u\bar{v}}}(x, y)\dif x\dif y
\] 
for $\alpha_{uv}>0$. 
Although we still do not have a closed form expression for the $P_{\text{c}, uv}$'s, their values are now easily computable. With the values of $P_{\text{c}, uv}$, the decoding error probability can be found via (\ref{MAPPb0}). 

A special case where the error probability has a closed form expression is when $\gamma_{\phi}=0$ and the source messages are uniformly distributed, i.e., $p_1=p_2=1/2$. In this case, we have 
\begin{equation}
P_{\text{err}} = 1 - \left(1-Q\left(\frac{|S_{11}-S_{10}|}{2\sigma}\right)\right)\left(1-Q\left(\frac{|S_{21}-S_{20}|}{2\sigma}\right)\right)\rule[-1.8em]{0pt}{0pt}
\label{UniPe} 
\end{equation} 
where $Q(x) \triangleq \int_{x}^{\infty}(1/\sqrt{2\pi})\exp(t^2/2)\dif t$ is the Gaussian Q-function. 
This expression is in fact the symbol error rate of the rectangular four quadrature-amplitude-modulation (4-QAM) in AWGN channels for uniformly distributed source messages. 
This result is expected because the transmission of two orthogonal signals $S_{1U}$ and $S_{2V}$ over a GMAC under joint MAP decoding can be viewed as a one-sender $4$-QAM transmission over an AWGN channel. 
When the two bits of a $4$-QAM symbol are independent of each other, the decoding of $4$-QAM can be decomposed into two independent detections, one for each bit. 
The decoding is correct only if the detection of both bits are correct, thereby yielding the expression in (\ref{UniPe}). 

\section{Optimized Design of Binary Constellations for Two-Sender GMAC}
For given sender waveforms $\phi_1(t)$ and $\phi_2(t)$, our objective is to find the coefficients $\{a_{jb}: j=1, 2, b=0, 1\}$ that minimize $P_{\text{err}}$. 
From a signal space viewpoint, this is equivalent to designing two BPAM constellations on the complex plane. In what follows, for the sake of completeness, the designs for all possible values of $\gamma_{\phi}$ are considered. 
Although for the case of $\gamma_{\phi}=0$, the optimized constellations were already derived in \cite{Tyson15}, here we present a simpler way to arrive at the same conclusion. 
For $\gamma_{\phi}=\pm 1$, we derive the optimal constellations based on minimizing the error rate under joint MAP decoding in the high SNR regime. 
For other values of $\gamma_{\phi}$, since the expression of the exact error rate is generally too complex, the union bound on the error rate in the high SNR regime is minimized. 
Specifically, let $P_{\text{err}}(\mathcal{S}_1, \mathcal{S}_2, \sigma^2)$ denote the system's error rate when the constellations $\mathcal{S}_1$ and $\mathcal{S}_2$ are employed and the noise variance is $\sigma^2$. 
For $\gamma_{\phi}\in\{0, 1, -1\}$, we determine constellations $\mathcal{S}_{1}$ and $\mathcal{S}_{2}$ such that 
\begin{IEEEeqnarray}{c}
\lim_{\sigma^2\rightarrow 0}\frac{P_{\text{err}}(\mathcal{S}_1, \mathcal{S}_2, \sigma^2)}{P_{\text{err}}(\mathcal{\tilde{S}}_1, \mathcal{\tilde{S}}_2, \sigma^2)}\le 1
\label{opt0}
\end{IEEEeqnarray}
for any other constellations $\mathcal{\tilde{S}}_1$ and $\mathcal{\tilde{S}}_2$. 
For other values of $\gamma_{\phi}$, the constellations $\mathcal{S}_{1}$ and $\mathcal{S}_{2}$ are optimized in the sense that  
\begin{IEEEeqnarray}{c}
\lim_{\sigma^2\rightarrow 0}\frac{P^{(\text{UB})}_{\text{err}}(\mathcal{S}_1, \mathcal{S}_2, \sigma^2)}{P^{(\text{UB})}_{\text{err}}(\mathcal{\tilde{S}}_1, \mathcal{\tilde{S}}_2, \sigma^2)}\le 1
\label{opt2}
\end{IEEEeqnarray}
for any other constellations $\mathcal{\tilde{S}}_1$ and $\mathcal{\tilde{S}}_2$, where $P_{\text{err}}^{(\text{UB})}$ denotes the union bound on the error rate.

\begin{lemma}
For $j=1, 2$, the two signal points of $\mathcal{S}_j$ which are separated by the largest possible distance under the average energy constraint given in (\ref{PowerConst}) are in the form of $S_{j0}=-\sqrt{(1-p_j)E_j/p_j}e^{i\eta}$ and $S_{j1}=\sqrt{p_jE_j/(1-p_j)}e^{i\eta}$, where $\eta\in[0, 2\pi)$ and the largest separation distance is given by $d_{j, \max}\triangleq\sqrt{E_j/(p_j(1-p_j))}$. 
\label{MAXSepConstellation}
\end{lemma}
\begin{proof}
Finding the signal points $S_{j0}$ and $S_{j1}$ which simultaneously achieve the largest separation distance and satisfy the average energy constraint is equivalent to solving the constrained quadratic minimization problem \cite{Ipatov07}: 
\[
\max \bm{x}\bm{\Sigma}\bm{x}^T\\
\text{subject to}\ ||\bm{x}||^2=1,
\]
where $(\cdot)^T$ and $||\cdot||$ respectively denote the transposition operation and the Euclidean norm, $\bm{x}=(\sqrt{p_j/E_j}\Re[S_{j0}], \allowbreak\sqrt{(1-p_j)/E_j}\Re[S_{j1}], \sqrt{p_j/E_j}\Im[S_{j0}], \allowbreak\sqrt{(1-p_j)/E_j}\Im[S_{j1}])$ is a row vector in which $\Im[z]$ denotes the imaginary part of the complex number $z$, and the $4\times 4$ positive semidefinite matrix $\bm{\Sigma}$ is given by
\begin{IEEEeqnarray}{c}
\Sigma=
\scalemath{0.9}{
\begin{pmatrix}[1]
\frac{E_j}{p_j} & -\frac{E_j}{\sqrt{p_j(1-p_j)}} & 0 & 0\\
-\frac{E_j}{\sqrt{p_j(1-p_j)}}  & \frac{E_j}{1-p_j} & 0 & 0\\
0&0&\frac{E_j}{p_j} & -\frac{E_j}{\sqrt{p_j(1-p_j)}}\\
0&0&-\frac{E_j}{\sqrt{p_j(1-p_j)}}  & \frac{E_j}{1-p_j}\\
\end{pmatrix}}.\nonumber
\end{IEEEeqnarray}
It is easy to verify that the $\bm{x}\bm{\Sigma}\bm{x}^T=|S_{j0}-S_{j1}|^2$, and the constraint $||\bm{x}||^2=1$ represents the average energy constraint (\ref{PowerConst}) for sender $j$. 
Using this formulation, the largest squared Euclidean distance and the corresponding signal points can be immediately obtained by determining the largest eigenvalue of $\bm{\Sigma}$ and the associated eigenvector \cite{Horn12}. 
These are respectively given by $E_j/(p_j(1-p_j))$ and $(-(1-p_j)E_j/p_j, p_jE_j/(1-p_j), 0, 0)$, yielding the signal points $S_{j0}=-\sqrt{(1-p_j)E_j/p_j}$ and $S_{j1}=\sqrt{p_jE_j/(1-p_j)}$ with the distance between them given by $\sqrt{E_j/(p_j(1-p_j))}\triangleq d_{j, \max}$. 
Moreover, since any constellation obtained by rotating $\mathcal{S}_j$ through an angle $\eta$ about the origin has the same separation distance and also satisfies the  energy constraint, the desired binary signals have the general form $S_{j0}=-\sqrt{(1-p_j)E_j/p_j}e^{i\eta}$ and $S_{j1}=\sqrt{p_jE_j/(1-p_j)}e^{i\eta}$, where $\eta\in [0, 2\pi)$. 
\end{proof}

We remark that the binary constellation given in Lemma~\ref{MAXSepConstellation} is in fact the optimal binary constellation for a single sender system with a non-uniformly distributed source \cite{Korn03}. 

We now give a few definitions.
When the constellations designed for two senders are constructed by only considering the marginal probabilities $p_1$ and $p_2$, the design is called an {\em individually optimized design}. 
If the constellation design is based on the joint probability distribution $p_{UV}$, we call it a {\em jointly optimized design}. 
Based on the optimality criteria presented in (\ref{opt0}) and (\ref{opt2}), we next derive the jointly optimized constellations for $\gamma_\phi=0$, $\gamma_\phi=\pm 1$, and other values of $\gamma_{\phi}$. 

\subsection{Design of Signal Constellations for $\gamma_{\phi}=0$}{\label{III-O}}
\begin{theorem}
The optimized constellations for the orthogonal transmission in the high SNR regime in the sense of (\ref{opt0}) are given by
\begin{equation}
S_{10}=-\sqrt{\frac{1-p_1}{p_1}E_1},\ S_{11}=\sqrt{\frac{p_1}{1-p_1}E_1}
\label{OptS1}
\end{equation}
and
\begin{equation}
S_{20}=-i\sqrt{\frac{1-p_2}{p_2}E_2},\ S_{21}=i\sqrt{\frac{p_2}{1-p_2}E_2}.
\end{equation}
\label{thm0}
\end{theorem}

To prove this theorem which recovers the result of \cite{Tyson15}, we need the following lemma. Here, $\hat{P}_{\text{err}}(\mathcal{S}_1, \mathcal{S}_2, \sigma^2)$ denotes the right-hand-side of (\ref{UniPe}). 

\begin{lemma}
For given constellations $\mathcal{S}_1=\{S_{10}, S_{11}\}$ and $\mathcal{S}_2=\{S_{20}, S_{21}\}$ with $\gamma_{\phi}=0$ and any source distribution $P_{UV}$, the error probability under joint MAP decoding is asymptotically given by (\ref{UniPe}) as $\sigma^2\rightarrow 0$, i.e., $\lim\limits_{\sigma^2\rightarrow 0} \hat{P}_{\text{err}}(\mathcal{S}_1, \mathcal{S}_2, \sigma^2)/P_{\text{err}}(\mathcal{S}_1, \mathcal{S}_2, \sigma^2)=1$.  
\label{lemma0}
\end{lemma}
\begin{proof}
If $\gamma_{\phi}=0$, then $\gamma=0$ because of $\zeta=0$ (see (\ref{gamma})) and the joint pdf $f_{\Delta_{uv, \bar{u}v}, \Delta_{uv, u\bar{v}}}(x, y)$ can be written as a product of two Gaussian density functions. When $\sigma^2\rightarrow 0$, we further have $\alpha_{uv}\rightarrow 0$ and $\beta\rightarrow 0$ for all $(u, v)\in\{0, 1\}^2$. 
It is easy to see that for any $(u, v)$ 
\begin{IEEEeqnarray}{rCl}
\lim_{\sigma^2\rightarrow 0}\frac{1}{P_{\text{c}, uv}}\left(1-Q\left(\frac{|S_{11}-S_{10}|}{2\sigma}\right)\right)\left(1-Q\left(\frac{|S_{21}-S_{20}|}{2\sigma}\right)\right)=1, 
\label{CROrthogonal}
\end{IEEEeqnarray}
which implies the lemma. 
\end{proof}

The result presented in Lemma~\ref{lemma0} is intuitively clear since joint MAP decoding reduces to maximum likelihood decoding for high SNRs, where the decoding performance is mainly dominated by the distance between signal points independently of the source distribution. 
Based on the two previous lemmas, we now prove Theorem \ref{thm0}.

\begin{proof}[Proof of Theorem \ref{thm0}]
By Lemma~\ref{lemma0}, to minimize the error probability in the high SNR regime, the magnitudes of $S_{11}-S_{10}$ and $S_{21}-S_{20}$ in (\ref{UniPe}) should be as large as possible. Under the energy constraint (\ref{PowerConst}), the signal points with the largest separation distance are given in Lemma~\ref{MAXSepConstellation}. 
By respectively choosing $\eta=0$ and $\eta=\pi/2$ for sender $1$ and sender $2$, the optimized constellations are obtained. 
\end{proof}

\begin{table*}[t]
\centering
\caption{The probability of correct decoding for $\gamma_{\phi}=\pm 1$ in the high SNR regime.}
\label{table1}
\begin{tabular}{|c|c|c|c|l|}
\cline{1-5}
Case & $d_1$ & $d_2$ & $|d_1|\lessgtr|d_2|$ & \multicolumn{1}{c|}{$\tilde{P}_{\text{c}}$}                                                                      \\ \hline
I & $>0$  & $>0$  & $>$                  & $1-Q\left(\frac{d_2}{2\sigma}\right)-(p_{10}+p_{01})Q\left(\frac{d_1-d_2}{2\sigma}\right)$               \\ \hline
II & $>0$  & $>0$  & $<$                  & $1-Q\left(\frac{d_1}{2\sigma}\right)-(p_{10}+p_{01})Q\left(\frac{d_2-d_1}{2\sigma}\right)$               \\ \hline
III & $>0$  & $<0$  & $>$                  & $Q\left(\frac{d_2}{2\sigma}\right)-(p_{00}+p_{11})Q\left(\frac{d_1+d_2}{2\sigma}\right)$                 \\ \hline
IV & $>0$  & $<0$  & $<$                  & $(p_{10}+p_{01})-Q\left(\frac{d_1}{2\sigma}\right)+(p_{00}+p_{11})Q\left(\frac{d_1+d_2}{2\sigma}\right)$ \\ \hline
V & $<0$  & $>0$  & $>$                  & $(p_{10}+p_{01})-Q\left(\frac{d_2}{2\sigma}\right)+(p_{00}+p_{11})Q\left(\frac{d_1+d_2}{2\sigma}\right)$ \\ \hline
VI & $<0$  & $>0$  & $<$                  & $Q\left(\frac{d_1}{2\sigma}\right)-(p_{00}+p_{11})Q\left(\frac{d_1+d_2}{2\sigma}\right)$                 \\ \hline
VII & $<0$  & $<0$  & $>$                  & $Q\left(\frac{d_2}{2\sigma}\right)-(p_{10}+p_{01})Q\left(\frac{d_2-d_1}{2\sigma}\right)$                 \\ \hline
VIII & $<0$  & $<0$  & $<$                  & $Q\left(\frac{d_1}{2\sigma}\right)-(p_{10}+p_{01})Q\left(\frac{d_1-d_2}{2\sigma}\right)$                 \\ \hline
\end{tabular}
\end{table*}

\subsection{Design of Signal Constellations for $\gamma_{\phi}=\pm 1$}{\label{III-B}}
Given signal constellations $\mathcal{S}_1$ and $\mathcal{S}_2$ (which lie on the real line in this case) and source distribution $p_{UV}$, the decision region for each message pair on the combined constellation can be identified and the MAP decoding performance can be readily evaluated. 
In the Appendix, we provide an example to demonstrate this procedure, in which certain set of conditions such as $D_{10, 00}>0$, $D_{01, 00}>0$, and $|D_{10, 00}|-|D_{01, 00}|>0$ are imposed on the signal points in order to explicitly characterize the decision region $\Omega$ given in (\ref{deltaeq}). 
However, to design optimal constellations, all such possible conditions on the signs of $D_{10, 00}$, $D_{01, 00}$, and $|D_{10, 00}|-|D_{01, 00}|$ must be taken into account. 
According to the relative position of $S_{1u}$ and $S_{2v}$, $u, v\in\{0, 1\}$, and the distance between them, there are eight cases that need to be considered. 
Based on (\ref{deltaeq}), each of these cases will lead to a different decision region for which we derive the MAP decoding performance in closed form. 
Then the case that achieves the minimum error rate is chosen as the optimal design. 
 
Although the above approach is straightforward, we note that tedious numerical computations and comparisons are required to obtain the optimal constellation.
To avoid designing signal constellations numerically, we construct constellations under the high SNR assumption as in the case $\gamma_{\phi}=0$. 
In this way, an explicit construction of the constellation is obtained, and we will see later that such a construction results in an negligible performance degradation relative to the truly optimal construction. 

For $j=1, 2$, let $d_j\triangleq S_{j1}-S_{j0}$. Here, $d_j$ is real-valued with $|d_j|\in(0, d_{j, \max}]$, where $d_{j, \max}$ is given in Lemma~\ref{MAXSepConstellation}. 
Without loss of generality, we assume that $p_1$, $p_2$, $E_1$, and $E_2$ are such that $d_{2, \max}\le d_{1, \max}$. 
For the case of $d_{2, \max}> d_{1, \max}$, we only need to swap the role of the two senders in the main result. 

\begin{theorem}
Suppose that $d_{2, \max}\le d_{1, \max}$. For $\gamma_\phi=\pm 1$ and high SNRs, the optimized constellation (in the sense of (\ref{opt0})) for sender $1$ is given by (\ref{OptS1}), while the optimized constellation for sender $2$ is given by
\begin{IEEEeqnarray}{c}
S_{20}=-\sqrt{\frac{1-p_2}{p_2}E_2},\ S_{21}=\sqrt{\frac{p_2}{1-p_2}E_2}\nonumber
\end{IEEEeqnarray}
if $|d_2|\ge d_{2, \max}$, and otherwise we have
\begin{IEEEeqnarray}{c}
S_{20}=S_{21}-d_2,\ S_{21}=d_2p_2\pm\sqrt{(d_2)^2p_2(p_2-1)+E_2}\nonumber
\end{IEEEeqnarray}
where
\[
d_2=\left\{
\begin{array}{l}
-4\sigma^2\ln(p_{10}+p_{01})/d_{1, \max}+d_{1, \max}/2,\ \text{if}\ (p_{00}+p_{11}) \ge (p_{10}+p_{01}),\\ 
4\sigma^2\ln(p_{11}+p_{00})/d_{1, \max}-d_{1, \max}/2,\ \text{if}\ (p_{00}+p_{11}) < (p_{10}+p_{01}),
\end{array}\right.
\]
with $d_{1, \max}\allowbreak=\sqrt{E_1/(p_1(1-p_1))}$. 
\label{thm1}
\end{theorem}

To prove this theorem, for each case we first find a closed form expression of the system's correct decoding probability $\tilde{P}_{\text{c}}$ in the high SNR regime. The conditions on the signal points for all eight cases are listed in Table~\ref{table1}. 
For Case~I, $\tilde{P}_{\text{c}}$ is obtained using the derivation in the Appendix by applying the high SNR assumption. 
The other $\tilde{P}_\text{c}$'s are derived in a similar fashion. 
From Table~\ref{table1}, we observe that by symmetry some cases can be disregarded without degrading our design. We need the following lemmas to simplify the procedure. 

\begin{lemma}
The maximum of $\tilde{P}_{\text{c}}(\text{Case VII})$ is the same as the maximum of $\tilde{P}_{\text{c}}(\text{Case I})$.  
\label{exclude1}
\end{lemma}
\begin{proof}
For Case~VII, by defining $\bar{d}_j=-d_j$ for $j=1, 2$, the correct decoding probability for high SNRs can be expressed as
\begin{IEEEeqnarray}{lCl}
\tilde{P}_{\text{c}}(\text{Case VII}) &=& Q\left(\frac{-\bar{d}_2}{2\sigma}\right)-(p_{01}+p_{10})Q\left(\frac{\bar{d}_1-\bar{d}_2}{2\sigma}\right)\nonumber \\
&=& \nonumber 1-Q\left(\frac{\bar{d}_2}{2\sigma}\right)-(p_{01}+p_{10})Q\left(\frac{\bar{d}_1-\bar{d}_2}{2\sigma}\right),\nonumber
\end{IEEEeqnarray}
where $\bar{d}_j\in (0, d_{j, \max}]$. 
The above expression reduces to $\tilde{P}_\text{c}(\text{Case I})$ after the substitutions $d_1\rightarrow\bar{d}_1$ and $d_2\rightarrow\bar{d}_2$. Moreover, the domain of the parameters are also the same, i.e., $d_j, \bar{d}_j\in (0, d_{j, \max}]$, $j=1, 2$.  Therefore, for any noise variance $\sigma^2$, we have
\[
\max_{d_1, d_2} \tilde{P}_{\text{c}}(\text{Case I})=\max_{\bar{d}_1, \bar{d}_2} \tilde{P}_{\text{c}}(\text{Case VII}). 
\]
\end{proof}

The argument of Lemma~\ref{exclude1} can be extended to Cases~II and VIII, Cases~III and V, and Cases~IV and VI. 
Since the corresponding proofs are almost identical, we omit the details. 
Based on these results, we can exclude Cases IV, V, VII, and VIII from consideration. 

\begin{lemma}
The maximum of $\tilde{P}_{\text{c}}(\text{Case II})$ cannot be less than the maximum of $\tilde{P}_{\text{c}}(\text{Case I})$. 
\label{CaseIIBad}
\end{lemma}
\begin{proof}
Define 
\begin{equation}
G(y_1, y_2) = 1 - Q\left(\frac{y_1}{2\sigma}\right) - (p_{10}+p_{01})Q\left(\frac{y_2}{2\sigma}\right)\nonumber
\end{equation}
for $y_1, y_2 >0$. 
For the maximum of correct decoding we have
\begin{equation}
\tilde{P}^*_{\text{c}}(\text{Case I})\triangleq\max_{y_1\in(0, d_{2, \max}]} \max_{y_2\in(0, d_{1, \max}-y_1]}G(y_1, y_2)\nonumber
\end{equation}
and 
\begin{equation}
\tilde{P}^*_{\text{c}}(\text{Case II})\triangleq\max_{y_1\in(0, d_{2, \max}]} \max_{y_2\in(0, d_{2, \max}-y_1]}G(y_1, y_2). \nonumber
\end{equation}
Since $d_{1, \max}\ge d_{2, \max}$, we obtain $\tilde{P}^*_{\text{c}}(\text{Case I})\ge \tilde{P}^*_{\text{c}}(\text{Case II})$. 
\end{proof}

By a similar argument, we can also show that $\tilde{P}^{*}_{\text{c}}(\text{Case VI})\allowbreak\le \tilde{P}^{*}_{\text{c}}(\text{Case III})$. Consequently, it suffices to design the optimized constellations for Cases~I and III, and the design with the larger $\tilde{P}_{\text{c}}$ of the two cases is the best design. 
The next two lemmas help us explicitly derive the optimized constellations for these two cases.

\begin{lemma}
For any fixed $d_2\in(0, d_{2, \max}]$, $\tilde{P}_{\text{c}}(\text{Case I})$ is increasing in $d_1$. 
\label{property1}
\end{lemma}  
\begin{proof}
Taking the partial derivative of $\tilde{P}_{\text{c}}(\text{Case I})$ with respect to $d_1$ yields  
\begin{equation}
\frac{\partial \tilde{P}_{\text{c}}(\text{Case I})}{\partial d_1}=(p_{01}+p_{10})\frac{1}{\sqrt{2\pi}\sigma}\exp(\frac{-(d_1-d_2)^2}{8\sigma^2})> 0. \nonumber
\end{equation}
\end{proof}

\begin{lemma}
$\tilde{P}_{\text{c}}(\text{Case I})$ is a concave function in $d_2$ for $d_1=d_{1, \max}$.
\label{property2} 
\end{lemma}

\begin{proof}
By taking partial derivatives of $\tilde{P}_{\text{c}}(\text{Case I})$ with respect to $d_2$, we have
\begin{IEEEeqnarray}{rCl}
\frac{\partial \tilde{P}_{\text{c}}(\text{Case I})}{\partial d_2}&=&\frac{1}{2\sqrt{2\pi\sigma^2}}\left(\exp\left(\frac{-d_2^2}{8\sigma^2}\right)\right.\left.-(p_{01}+p_{10})\exp\left(\frac{-(d_{1, \max}-d_2)^2}{8\sigma^2}\right)\right)\nonumber
\label{FirstCond}
\end{IEEEeqnarray}
and
\begin{IEEEeqnarray}{l}
\frac{\partial^2 \tilde{P}_{\text{c}}(\text{Case I})}{\partial d^2_2}=\frac{1}{16\sigma^2\sqrt{2\pi}}\left(\exp\left(\frac{-d_2^2}{8\sigma^2}\right)(-2d_2)\right.\left. -(p_{01}+p_{10})\exp\left(\frac{-(d_{1, \max}-d_2)^2}{8\sigma^2}\right)2(d_{1, \max}-d_2)\right).\nonumber\\ \*
\label{SecondCond}
\end{IEEEeqnarray}
Since $d_{2, \max}\le d_{1, \max}$, the second derivative given in (\ref{SecondCond}) is non-positive for all $d_2\in (0, d_{2, \max}]$. 
Hence, $\tilde{P}_{\text{c}}(\text{Case I})$ is a concave function of $d_2$ on the interval $(0, d_{2, \max}]$. 
\end{proof}

Similarly, one can show that for any $d_1\in(0, d_{2, \max}]$, $P_{\text{c}}(\text{Case III})$ is increasing in $d_2$, where $d_1$ is upper bounded by $d_{2, \max}$ due to the condition $|d_1|<|d_2|$. 
Also, for $d_2=d_{2, \max}$, $\tilde{P}_\text{c}(\text{Case III})$ is a concave function in $d_1$ for $d_1\in(0, d_{2, \max}]$, 
Since the proofs of these statements are almost identical to the proof for Case~I, the details are omitted. 
Based on the previous lemmas, we are readily to prove Theorem~\ref{thm1}. 

\begin{proof}[Proof of Theorem \ref{thm1}]
First, we show that if $p_{00}+p_{11}\ge p_{10}+p_{01}$, then $\tilde{P}^*_{\text{c}}(\text{Case I})\ge \tilde{P}^*_{\text{c}}(\text{Case III})$, which implies that we only need to consider Case~I. By letting $\bar{d}_2=-d_2$, we can rewrite $\tilde{P}_{\text{c}}(\text{Case III})$ as
\begin{IEEEeqnarray}{lCl}
\tilde{P}_{\text{c}}(\text{Case III})&=& Q\left(\frac{-\bar{d}_2}{2\sigma}\right)-(p_{00}+p_{11})Q\left(\frac{d_1-\bar{d}_2}{2\sigma}\right)\nonumber \\
&=& 1-Q\left(\frac{\bar{d}_2}{2\sigma}\right)-(p_{00}+p_{11})Q\left(\frac{d_1-\bar{d}_2}{2\sigma}\right),\label{newPc3}
\end{IEEEeqnarray}
where $d_1\in(0, d_{1, \max}]$ and $\bar{d}_2\in(0, d_{2, \max}]$. 
Based on this expression and the correspondence $d_2\leftrightarrow \bar{d}_2$, the feasible set for $\tilde{P}_{\text{c}}(\text{Case III})$ is observed to be identical to that for $\tilde{P}_{\text{c}}(\text{Case I})$. Moreover, the new expression for $\tilde{P}_{\text{c}}(\text{Case III})$ only differs from $\tilde{P}_{\text{c}}(\text{Case I})$ in the coefficient of the third term. 
Therefore, when $p_{00}+p_{11}\ge p_{10}+p_{01}$, we have $\max_{d_1, d_2}\tilde{P}_{\text{c}}(\text{Case I})\ge \max_{d_1, \bar{d}_2}\tilde{P}_{\text{c}}(\text{Case III})$ and the optimized constellations for Case I should be selected. 
In contrast, when $p_{00}+p_{11}< p_{10}+p_{01}$, the optimized constellations for Case~III is chosen. 
Next, we explicitly derive the optimal constellations for Cases~I and III to complete the proof. 

To find the optimized constellation $\mathcal{S}_1$ that maximizes $\tilde{P}_{\text{c}}(\text{Case I})$, $d_1$ should be set to its maximum possible value $d_1=d_{1, \max}$ according to Lemma~\ref{property1}. 
By setting $\eta=0$ in Lemma~\ref{MAXSepConstellation}, this choice of $d_1$ immediately gives the optimal constellation for sender $1$:   
\begin{equation}
S_{10}=-\sqrt{\frac{1-p_1}{p_1}E_1},\ S_{11}=\sqrt{\frac{p_1}{1-p_1}E_1}.
\label{OptConstellation1}  
\end{equation}
Moreover, based on the concavity property in Lemma~\ref{property2}, the maximum of $\tilde{P}_{\text{c}}(\text{Case I})$ in the variable $d_2$ is known to occur at either where the partial derivative is zero or at the boundary of its support interval. Solving $\frac{{\partial \tilde{P}_{\text{c}}}(\text{Case I})}{{\partial d_2}}=0$ for $d_2$, we obtain
\begin{equation}
d_2=-\frac{4\sigma^2}{d_{1, \max}}\ln(p_{10}+p_{01})+\frac{d_{1, \max}}{2}.
\label{OptDistd2}
\end{equation}
By substituting $S_{21}-S_{20}=d_2$ into the average energy constraint $p_2|S_{20}|^2+(1-p_2)|S_{21}|^2=E_2$, the optimized $\mathcal{S}_2$ is obtained. 
In summary, when $d_2^2p_2(p_2-1)+E_2> 0$ with $d_2$ given in (\ref{OptDistd2}), there are two optimized constellations: 
\begin{equation}
S_{20} = S_{21}-d_2,\ S_{21}=d_2p_2+\sqrt{d_2^2p_2(p_2-1)+E_2}\nonumber
\label{OptConstellation2_1}
\end{equation}
and 
\begin{equation}
S_{20} = S_{21}-d_2,\ S_{21}=d_2p_2-\sqrt{d_2^2p_2(p_2-1)+E_2}. \nonumber
\label{OptConstellation2_11}
\end{equation}
When $d_2^2p_2(p_2-1)+E_2\le 0$, the optimized constellation for sender $2$ is
\begin{equation}
S_{20}=-\sqrt{\frac{1-p_2}{p_2}E_2},\ S_{21}=\sqrt{\frac{p_2}{1-p_2}E_2}\nonumber
\label{OptConstellation2_2}
\end{equation}
which follows from the result that the maximum of $\tilde{P}_\text{c}(\text{Case I})$ occurs at $d_2=d_{2, \max}$ and from Lemma~\ref{MAXSepConstellation}.

With the help of the expression in (\ref{newPc3}) and the above derivation, the optimal constellations for Case~III can be easily derived. 
For sender $1$, the optimized constellation $\mathcal{S}_1$ is the same as the one given in (\ref{OptConstellation1}) because the choice $d_1=d_{1, \max}$ also maximizes $\tilde{P}_{\text{c}}(\text{Case III})$. 
For sender $2$, solving $\frac{{\partial \tilde{P}_{\text{c}}}(\text{Case III})}{{\partial \bar{d}_2}}=0$ gives
\begin{equation}
\bar{d}_2=-\frac{4\sigma^2}{d_{1, \max}}\ln(p_{00}+p_{11})+\frac{d_{1, \max}}{2}. \nonumber
\label{OptDistd2CaseIII}
\end{equation}
By substituting $S_{20}-S_{21}=\bar{d}_2$ into $p_2|S_{20}|^2+(1-p_2)|S_{21}|^2=E_2$, there are two optimized constellations $\mathcal{S}_2$ in case of $\bar{d}_2^2p_2(p_2-1)+E_2> 0$ given by 
\begin{equation}
S_{20} = S_{21}+\bar{d}_2,\ S_{21}=\bar{d}_2p_2+\sqrt{\bar{d}_2^2p_2(p_2-1)+E_2}\nonumber
\label{OptConstellation2_12}
\end{equation}
and 
\begin{equation}
S_{20} = S_{21}+\bar{d}_2,\ S_{21}=\bar{d}_2p_2-\sqrt{\bar{d}_2^2p_2(p_2-1)+E_2}.\nonumber
\label{OptConstellation2_13}
\end{equation}
When $\bar{d}_2^2p_2(p_2-1)+E_2\le 0$, the optimized constellation $\mathcal{S}_2$ is given by 
\begin{IEEEeqnarray}{c}
S_{20}=\sqrt{\frac{1-p_2}{p_2}E_2},\ S_{21}=-\sqrt{\frac{p_2}{1-p_2}E_2}.\nonumber
\label{OptConstellation2_3}
\end{IEEEeqnarray}  
\end{proof}

\subsection{Design of Signal Constellations for $\gamma_{\phi}\notin\{0, 1, -1\}$}
When $\gamma_{\phi}\notin\{0, 1, -1\}$, the design procedure becomes more difficult since the performance of joint MAP decoding is usually not known in closed form even at high SNRs. 
Instead, we use the union bound to facilitate the design procedure. 
Specifically, a closed form upper bound for the error probability of MAP decoding is first obtained via the union bound. 
The optimized constellations are then derived analytically based on the minimization of this upper bound in the high SNR regime.   
Although this design approach is less accurate than optimizing the exact system error rate (or errror bounds that are tighter than the union bound \cite{Kuai10}), its effectiveness has been extensively demonstrated in, e.g., \cite{Foschini74}-\cite{Maleki14}. 

The union bound on the error rate of the joint MAP decoder is
\begin{equation}
P_{\text{err}}\le P_{\text{err}}^{(\text{UB})}=\sum\limits_{(u, v)}\sum\limits_{(l,m)\neq(u, v)}p_{uv}\Pr(\Delta_{uv, lm}> 0). 
\label{ERUB}
\end{equation}
In the high SNR regime, we further have that 
\begin{equation}
\lim_{\sigma^2\rightarrow 0}\frac{Q\left(\frac{|A_{lm}-A_{uv}|}{2\sigma}\right)}{\Pr(\Delta_{uv, lm}> 0)}=1\nonumber
\label{PEP}
\end{equation} 
for all $(u, v)\neq(l, m)$, and hence the right-hand-side of (\ref{ERUB}) can be approximated for $\sigma^2$ sufficiently small by $\tilde{P}_{\text{err}}^{\text{(UB)}}$ given by
\begin{IEEEeqnarray}{l}
\tilde{P}_{\text{err}}^{\text{(UB)}}=Q\left(\frac{|d_1|}{2\sigma}\right) + Q\left(\frac{|d_2|}{2\sigma}\right) + (p_{00}+p_{11})Q\left(\frac{\sqrt{|d_1|^2+|d_2|^2+2|d_1||d_2|\cos\psi}}{2\sigma}\right) \nonumber\\ 
\quad\quad\quad\quad\quad\quad\quad + (p_{01}+p_{10})Q\left(\frac{\sqrt{|d_1|^2+|d_2|^2-2|d_1||d_2|\cos\psi}}{2\sigma}\right),\IEEEeqnarraynumspace
\label{HIGHSNRUB}
\end{IEEEeqnarray}
where $d_1=S_{11}-S_{10}$ and $d_2=S_{21}-S_{20}$ are generally complex-valued and $\psi$ is the angle measured counterclockwise from $d_1$ to $d_2$ on the complex plane. 
We note that taking partial derivatives to minimize (\ref{HIGHSNRUB}) results in transcendental equations, so another method is proposed here. 
Our design procedure is to first find optimal $|d_1|$, $|d_2|$, and $\psi$ that minimize (\ref{HIGHSNRUB}), and then derive optimized signals based on the energy constraints. 
We first note that for given basic waveforms with $\gamma_{\phi}=\cos\theta$, where $\theta$ denotes the angle between the signal subspaces induced by $\phi_1(t)$ and $\phi_2(t)$ on the complex plane, due to symmetry it is sufficient to consider $\theta\in(0, \pi)$. 
Also, by definition, $\psi$ can take value in $\{\theta, \theta+\pi\}$.  

Next, we give an example to illustrate our design approach.  
Suppose that $\gamma_{\phi}>0$ and $(p_{00}+p_{11})\ge(p_{01}+p_{10})$.
In this case, we have $0<\theta<\pi/2$. 
To minimize (\ref{HIGHSNRUB}), due to the possible values of $\psi$, the symmetry of the arguments of the $Q$ function in the third and forth terms of (\ref{HIGHSNRUB}), and $(p_{00}+p_{11})\ge(p_{01}+p_{10})$, we first choose $\psi=\theta$. 
Also, since the $Q$ function is decreasing and 
\begin{IEEEeqnarray}{c}
\lim_{\sigma^2\rightarrow 0}\frac{c^*Q(d^*)}{\tilde{P}_{\text{err}}^{\text{(UB)}}}=1,
\label{Dominant}
\end{IEEEeqnarray}
where $d^*$ denotes the minimum value among the arguments of the $Q$ function in $\tilde{P}_{\text{err}}^{\text{(UB)}}$, and $c^*$ is the associated coefficient of that term, we next maximize 
\begin{equation}
\min(|d_1|^2, |d_2|^2, |d_1|^2+|d_1|^2-2|d_1||d_2|\cos\psi)
\label{HIGHSNRDMIN}
\end{equation}
over $|d_1|\in(0, d_{1, \max}]$, $|d_2|\in(0, d_{2, \max}]$, and $d_{2, \max}\le d_{1,\max}$. 
Note that due to $|d_1|^2+|d_1|^2+2|d_1||d_2|\cos\psi\allowbreak>|d_1|^2+|d_1|^2-2|d_1||d_2|\cos\psi$, the term $|d_1|^2+|d_1|^2+2|d_1||d_2|\cos\psi$ has been excluded in (\ref{HIGHSNRDMIN}). 

To find $|d_1|$ and $|d_2|$ that maximize (\ref{HIGHSNRDMIN}), we use a two-step procedure.  
Roughly speaking, for an arbitrary but fixed $|d_1|$, we first identify the candidates for an optimal $|d_2|$ in (\ref{HIGHSNRDMIN}). 
These are either constants or functions of $|d_1|$.  
Using these candidates, we re-examine (\ref{HIGHSNRDMIN}) to find the optimal $|d_1|$. 
A few pairs of $|d_1|$ and $|d_2|$ which possibly maximize (\ref{HIGHSNRDMIN}) are then formed for further evaluation. 
We summarize the obtained results in the next lemma; the details of the two-steps procedure are given in the proof. 

\begin{figure*}[!t]
\centering
\subfloat[]{\includegraphics[draft=false, scale=0.44]{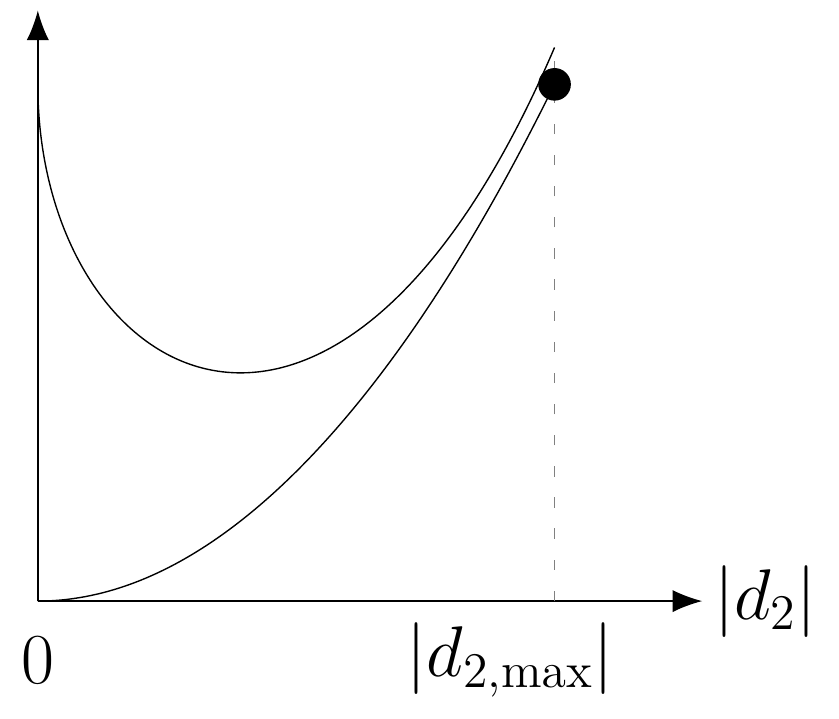}
\label{Ex4aFig}}\hspace{+0.5cm}
\subfloat[]{\includegraphics[draft=false, scale=0.44]{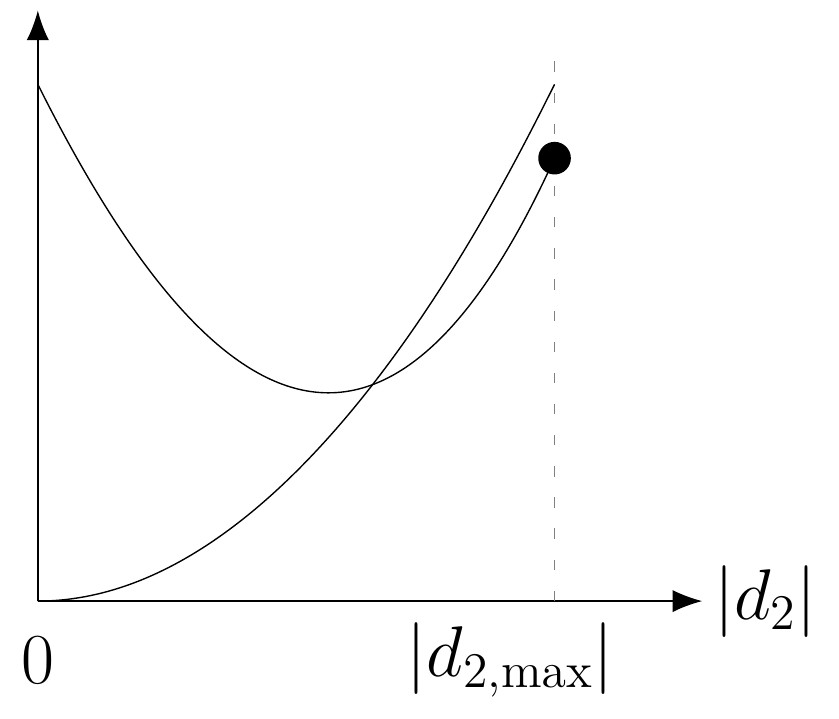}
\label{Ex4bFig}}\hspace{+0.5cm}
\subfloat[]{\includegraphics[draft=false, scale=0.44]{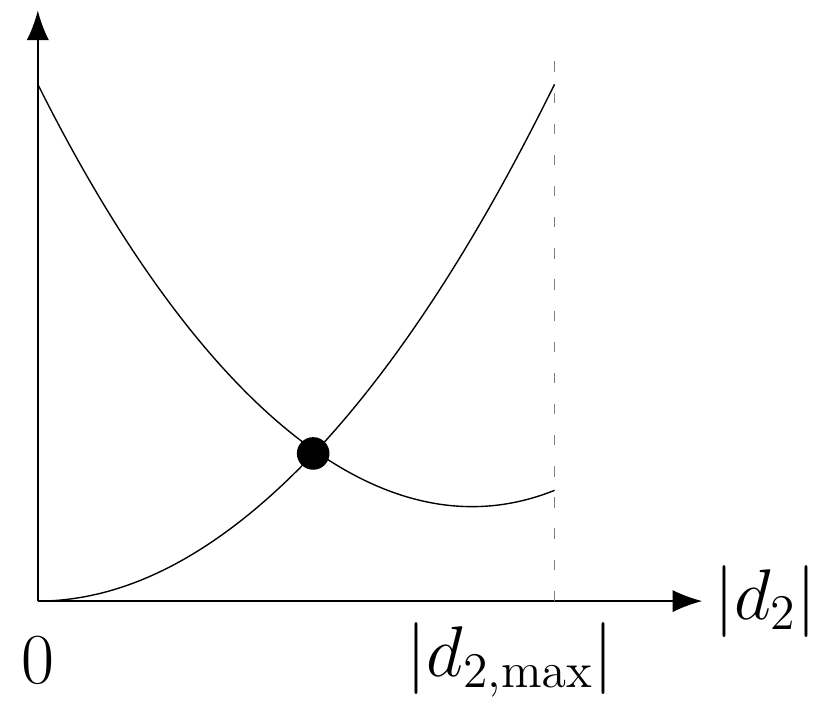}
\label{Ex4cFig}}
\caption{The possible locations of the maximum of $\min(|d_2|^2, |d_1|^2+|d_2|^2-2|d_1||d_2|\cos\psi)$, in which the maximum of (a) and (b) occurs at the boundary of support and the maximum of (c) happens at the intersection of the two curves.}
\label{Ex4}
\end{figure*}

\begin{lemma}
For any $0<\psi<\pi/2$, the optimal $|d_{1}|$ and $|d_{2}|$ that maximize (\ref{HIGHSNRDMIN}) are given by $|d_1|=d_{1, \max}$ and $|d_2|=d_{1, \max}/(2\cos\psi)$ if $d_{1, \max}^2+d_{2, \max}^2-2d_{1, \max}d_{2, \max}\cos\psi\le (d_{1, \max}/2\cos\psi)^2\le d_{2, \max}^2$, and $|d_2|=d_{2, \max}$ otherwise. 
\label{2DLemma} 
\end{lemma}
\begin{proof}
For the given $0<\psi<\pi/2$, we fix $|d_1|\in(0, d_{1, \max}]$ and define $w(|d_2|)=\min(|d_2|^2, |d_1|^2+|d_2|^2-2|d_1||d_2|\cos\psi)$. 
Note that $(\ref{HIGHSNRDMIN})$ can be now expressed as $\min(|d_1|^2, w(|d_2|))$. We first analyze $w(|d_2|)$. 
Since $|d_2|^2$ is increasing in $|d_2|$ and $|d_1|^2+|d_2|^2-2|d_1||d_2|\cos\psi$ is a quadratic function of $|d_2|$ on $(0, d_{2, \max}]$, the maximum value of $w(|d_2|)$ appears at either the interval boundary or at the intersection of the two curves. 
These cases are illustrated in Fig.\ \ref{Ex4}. 
For cases (a) and (b), the optimal $|d_2|$ is simply $d_{2, \max}$. 
For case (c), the optimal $|d_2|$ can easily be found by solving the equation $|d_1|^2+|d_2|^2-2|d_1||d_2|\cos\psi=|d_2|^2$ for $|d_2|$, which immediately gives $|d_2|=|d_1|/(2\cos\psi)$. Two possible candidates for optimal $|d_2|$ are then obtained. 
We note that the method of finding $|d_2|$ in the latter case does not take the domain of $|d_2|$ into account so that we have to check whether or not the obtained $|d_2|$ is on $(0, d_{2, \max}]$.  
If not, then this case degenerates to the former one and we have $|d_2|=d_{2, \max}$. 
With the two candidates for the optimal $|d_2|$, the function $\min(|d_1|^2, w(|d_2|))$ can be rewritten as 
\begin{IEEEeqnarray}{l}
\min\left(|d_1|^2, \frac{|d_1|^2}{(2\cos\psi)^2}, \right.|d_1|^2+d_{2, \max}^2-2|d_1|d_{2, \max}\cos\psi, d_{2, \max}^2\Bigg).\IEEEeqnarraynumspace 
\label{HIGHSNRDMIN2}
\end{IEEEeqnarray}
In (\ref{HIGHSNRDMIN2}), it is observed that the first three terms are increasing in $|d_1|$ and the fourth term is a constant so that we choose $|d_1|=d_{1, \max}$ to maximize the minimum value. 
Consequently, two possibly optimal pairs $(|d_1|, |d_2|)$ are formed. 
Moreover, to provide conditions for identifying the optimal pair, it suffices to find the case where the maximum value of (\ref{HIGHSNRDMIN2}) is achieved by $|d_1|=d_{1, \max}$ and $|d_2|=d_{1, \max}/(2\cos\psi)$. 
As illustrated in Fig. \subref*{Ex4cFig}, finding the condition is straightforward. The intersection is located between the boundary values of the two curves at $|d_2|=d_{2, \max}$, i.e.,  
\begin{IEEEeqnarray}{l}
\ d_{1, \max}^2+d_{2, \max}^2-2d_{1, \max}d_{2, \max}\cos\psi\le \left(\frac{d_{1, \max}}{2\cos\psi}\right)^2 \le d^2_{2, \max}.\nonumber
\end{IEEEeqnarray}
Combining the above observations, the lemma follows.  
\end{proof}

We are now ready to derive the optimized constellations for the case $\gamma_{\phi}>0$ and $(p_{00}+p_{11})\ge(p_{01}+p_{10})$. 
By the choice $|d_1|=d_{1, \max}$, $\mathcal{S}_1$ is the same as in (\ref{OptS1}). 
For sender $2$, if the best choice of $|d_2|$ is $d_{2, \max}$, then by Lemma {\ref{MAXSepConstellation} the optimized constellation $\mathcal{S}_2$ is given by
\begin{equation}
S_{20}=-\sqrt{\frac{(1-p_2)E_2}{p_2}}e^{i\eta}\ \text{and}\ S_{21}=\sqrt{\frac{p_2E_2}{1-p_2}}e^{i\eta}\nonumber
\end{equation}
with $\eta=\psi$. 
In contrast, if the best choice of $|d_2|$ is $d_{1,\max}/(2\cos\theta)$, the optimized constellation of sender $2$ can be obtained by solving the equation $|S_{21}-S_{20}|=|d_2|$ under the average energy constraint. 
After some algebra, the optimized constellation is found to be 
\begin{equation}
S_{20}=S_{21}-|d_2|e^{i\psi}\nonumber
\label{2DS2I}
\end{equation}
and
\begin{equation}
S_{21}= \left(p_2|d_2|\pm\sqrt{p_2(p_2-1)|d_2|^2+E_2}\right)e^{i\psi}\nonumber
\label{2DS2II}
\end{equation}
where both choices of $S_{21}$ result in optimized signals. 
The results for the general cases are summarized in the following theorem.  
We remark that the above derivations rely on the assumption that $d_{2, \max}\le d_{1, \max}$.
There is no loss of generality because for $d_{2, \max}> d_{1, \max}$ we only need to swap the roles of senders~1 and~2 in Lemma~9 and the main theorem. 
The proofs are almost identical as before. 

\begin{theorem}
Suppose that $d_{2, \max}\le d_{1, \max}$. For given $\gamma_{\phi}=\cos\theta$ with $\theta\in(0, \pi)$, the optimized constellation for sender $1$ is the same as the one given by (\ref{OptS1}) in Theorem~\ref{thm1}. For sender $2$, the optimized constellation is given by
\begin{equation}
S_{20}=S_{21}-|d_2|e^{i\psi}\nonumber
\end{equation}
and
\begin{equation}
S_{21}= \left(p_2|d_2|\pm\sqrt{p_2(p_2-1)|d_2|^2+E_2}\right)e^{i\psi}\nonumber
\end{equation}
where the value of $\psi$ for the cases $p_{00}+p_{11} \ge p_{10}+p_{01}$ and $p_{00}+p_{11} < p_{10}+p_{01}$ is respectively $\theta$ and $\theta+\pi$ and $|d_2|=d_{1, \max}/|2\cos\psi|$ if $d_{1, \max}^2+d_{2, \max}^2-2d_{1, \max}d_{2, \max}|\cos\psi|\le (d_{1, \max}/|2\cos\psi|)^2\le d_{2, \max}^2$, and $|d_2|=d_{2, \max}$ otherwise. 
\label{thm2}
\end{theorem}

\begin{proof}[Proof of Theorem \ref{thm2}]
For the case $\gamma_{\phi}>0$ and $(p_{00}+p_{11})\ge(p_{01}+p_{10})$, the optimized constellations have been derived above. 
The same procedure can be applied to other cases. To shorten the proof, we only point out some key points when applying the same procedure to other cases. 
Taking the case that $\gamma_{\phi}>0$ and $(p_{00}+p_{11})<(p_{01}+p_{10})$ as an example, here, $\psi$ should set to be $\theta+\pi$ to minimize (\ref{HIGHSNRUB}). 
By this choice of $\psi$ and (\ref{Dominant}), we should maximize
\begin{equation}
\min(|d_1|^2, |d_2|^2, |d_1|^2+|d_1|^2+2|d_1||d_2|\cos\psi).\nonumber
\label{HIGHSNRDMIN1}
\end{equation}
By defining $\bar{\psi}=\psi-\pi$, we can rewrite $\cos\psi$ as $-\cos\bar{\psi}$ and obtain an expression identical to (\ref{HIGHSNRDMIN}). 
From that expression, based on our previous argument, the optimal pair $(|d_1|, |d_2|)$ and the optimized constellations can be easily obtained. Extending to the case of $\gamma_{\phi}<0$ is straightforward and hence we omit the details. 
\end{proof}

\begin{table*}[t!]
\centering
\caption{The BPAM constellations for sources with $\underline{p}_{UV,\text{Case}1}$, where $\gamma_{\phi}=1$, $E_1=E_2=1$, and SNR $=18$dB.}
\label{1DExample}
\begin{tabular}{l|l|l|l|}
\cline{2-4}
 & $[S_{10}, S_{11}]$ & $[S_{20}, S_{21}]$ & $[A_{00}, A_{01}, A_{10}, A_{11}]$ \\ \hline
\multicolumn{1}{|l|}{Conventional Antipodal BPAM}           & $[-1, 1]$          & $[-1, 1]$          & $[-2, 0, 0, 2]$                    \\ \hline
\multicolumn{1}{|l|}{Individually Optimized BPAM} & $[-3, 1/3]$        & $[-3, 1/3]$        & $[-6, -8/3, -8/3, 2/3]$            \\ \hline
\multicolumn{1}{|l|}{Jointly Optimized BPAM} & $[-3, 1/3]$        & $[-2.421, -0.678]$ & $[-5.421, -3.678, -2.088, -0.345]$ \\ \hline
\multicolumn{1}{|l|}{Numerically Optimized BPAM}  & $[-3, 1/3]$        & $[-2.401, -0.686]$ & $[-5.401, -3.686, -2.068, -0.353]$ \\ \hline
\end{tabular}
\end{table*}

\begin{table*}[ht!]
\centering
\caption{The BPAM constellations for sources with $\underline{p}_{UV,\text{Case}2}$, where $\gamma_{\phi}=1$, $E_1=E_2=1$, and SNR $=18$dB.}
\label{1DExample2}
\begin{tabular}{l|l|l|l|}
\cline{2-4}
 & $[S_{10}, S_{11}]$ & $[S_{20}, S_{21}]$ & $[A_{00}, A_{01}, A_{10}, A_{11}]$ \\ \hline
\multicolumn{1}{|l|}{Conventional Antipodal BPAM}           & $[-1, 1]$          & $[-1, 1]$          & $[-2, 0, 0, 2]$                    \\ \hline
\multicolumn{1}{|l|}{Individually Optimized BPAM} & $[-2, 0.5]$        & $[-1, 1]$        & $[-3, -1, -0.5, 1.5]$            \\ \hline
\multicolumn{1}{|l|}{Jointly Optimized BPAM} & $[-2, 0.5]$        & $[-1.408, -0.131]$ & $[-3.408, -2.131, -0.908, 0.369]$ \\ \hline
\multicolumn{1}{|l|}{Numerically Optimized BPAM}  & $[-2, 0.5]$        & $[-1.406, -0.151]$ & $[-3.406, -2.151, -0.906, 0.349]$ \\ \hline
\end{tabular}
\end{table*}

We close this section with some observations.

\vspace{+0.15cm}\noindent{\it Observation 1:} The optimized constellation for $\gamma_{\phi}=\pm 1$ can possibly be constructed using (\ref{opt2}) instead of (\ref{opt0}). 
For example, letting $\theta=0$ in Theorem \ref{thm2}, we observe that the optimized constellation based on (\ref{opt2}) is identical to that given in Theorem~\ref{thm1} for high SNRs. This result is expected because $\lim\limits_{\sigma^2\rightarrow 0}P_{\text{err}}(\mathcal{S}_1, \mathcal{S}_2, \sigma^2)/P^{(\text{UB})}_{\text{err}}(\mathcal{S}_1, \mathcal{S}_2, \sigma^2)=1$. However, since the error rate approximation $1-\allowbreak\tilde{P}_{\text{c}}(\mathcal{S}_1, \mathcal{S}_2, \sigma^2)$ is tighter than $P^{(\text{UB})}_{\text{err}}(\mathcal{S}_1, \mathcal{S}_2, \sigma^2)$ at high SNR, the optimized constellation given in Theorem~\ref{thm1} is better than that obtained from Theorem~\ref{thm2}. 

\vspace{+0.15cm}\noindent{\it Observation 2:} Unlike the $4$-PAM constellation designed for single sender AWGN channels with a non-uniformly distributed source, the signal points at the boundary position of the combined constellation may not have the highest probability. 
This is because the optimized constellation is not only designed to combat channel noise but also to mitigate user interference. 

\vspace{+0.15cm}\noindent{\it Observation 3:} When the two sources $U$ and $V$ are uniformly distributed, the optimal combined constellation for $\gamma_{\phi}=1$ depends on the SNR and their signals are not equally spaced. 
This is related to the fact that the conventional uniform $4$-PAM constellation for single sender AWGN channels is not optimal under the average energy constraint \cite{Mako06}. However, when the SNR increases, the combined constellation signals become equally spaced. 

\vspace{+0.15cm}\noindent{\it Observation 4:} When $E_1=E_2$ and the marginal probability distributions of the sources are very biased $(p_1 \ll p_2)$, the jointly optimized constellations are identical to the individually optimized constellations. 
This indicates that if $S_{10}$ and $S_{11}$ are separated by a large distance, the interference from sender $2$ becomes negligible. 
For example, suppose that $S_{10}$ and $S_{11}$ can be separated by $d_{1, \max}$ and $d_{2, \max}\ll d_{1, \max}$. 
In this case, the system's error rate is expected to be dominated by the distance $|d_2|$ between the two signal points of sender $2$. Hence, $S_{20}$ and $S_{21}$ should be separated by the largest possible distance $|d_2|=d_{2, \max}$ to lower the error rate, which results in the same constellation optimized only for $p_2$. 

\section{Simulation Results}
In this section, we evaluate the effectiveness of our constellation designs via simulations. 
We let $E_1=E_2=1$, and the SNR is defined as $(E_1+E_2)/N_0=2/N_0$, where $N_0=\sigma^2$ if $\gamma_{\phi}=\pm 1$ and $N_0=2\sigma^2$ otherwise. 
For performance comparison, the conventional antipodal BPAM is considered. 
For $\gamma_{\phi}= 1$, these BPAM signals simply correspond to the signal points $S_{10}=S_{20}=1$ and $S_{11}=S_{21}=-1$ on the complex plane.
The optimal constellation designed for a single sender AWGN channel with a non-uniform binary source is also included \cite{Korn03}. 
As before, we call such a constellation individually optimized. 
The decoding performance of numerically optimized constellations, i.e., constellations that minimize the exact system error rate and are obtained by exhaustive search, is also provided as a reference

We only present simulation results for $\gamma_{\phi}\neq 0$ because the constellation design for orthogonal transmission $(\gamma_{\phi}=0)$ was tackled in \cite{Tyson15}. 
The joint source distributions we consider are $\underline{p}_{UV, \text{Case}1}\triangleq [p_{00}, p_{01}, p_{10}, p_{11}]=[0.091, 0.009, 0.009, 0.891]$ with $p_1=p_2=0.1$ and $\underline{p}_{UV, \text{Case}2}=[0.18, 0.02, 0.32, 0.48]$ with $p_1=0.2$ and $p_2=0.5$.
Here, the first source distribution has a stronger correlation than the second. 
The individually optimized constellations are simply generated by Lemma~\ref{MAXSepConstellation}, while the jointly optimized constellations for $\gamma_{\phi}=\pm 1$ and $\gamma_{\phi}\neq 0, \pm 1$ are constructed using Theorems~\ref{thm1} and \ref{thm2}, respectively. 

\subsection{Results for $\gamma_{\phi}=\pm 1$}
We consider the case $\gamma_{\phi}=1$, which results in the strongest user interference. 
For the joint distribution $\underline{p}_{UV, \text{Case}1}$, the signal points for various BPAMs at SNR $= 18$dB are listed in Table.~\ref{1DExample}, and their decoding performance is shown in Fig.~\ref{1D090101}. 
From the simulation results, the conventional antipodal BPAM is observed to exhibit poor decoding performance. 
This is partly due to the fact that identical BPAM at both sender introduces an ambiguity for the transmitted signals, i.e., $A_{01}=A_{10}=0$ (recall that $A_{uv}=S_{1u}+S_{2v}$). The same is true for the individually optimized BPAM because of the identical marginal distributions. 
However, since the individually optimized constellation achieves a larger average separation distance between the combined signals, its decoding performance can be slightly better than the conventional BPAM. This result indicates that not all non-bijective combined constellation are equally bad. 
In contrast, the jointly optimized constellation derived from our analysis is shown to provide significant improvement. 
Compared with the performance of the numerically optimized constellations, the difference is negligible.  

For $\gamma_{\phi}=1$ and $\underline{p}_{UV, \text{Case}2}$, the constellations designed for SNR $=18$dB are listed in Table~\ref{1DExample2}. 
As already noted, using the conventional antipodal BPAM for $\gamma_{\phi}=1$ unavoidably leads to an ambiguity for the transmitted signals and is also sub-optimal for non-uniform message sources. 
Nevertheless, due to the distinct marginal distribution for $U$ and $V$, the individually optimized BPAM now has different constellations at the two transmitters, thereby resulting in some performance improvement.  
Our designed constellation still significantly outperforms the individually optimized BPAM. 
The decoding performance of our designed constellation and the numerically optimized constellation are also nearly identical. 

\subsection{Results for $\gamma_{\phi}\notin\{0, 1, -1\}$}
We next present simulation results for $\gamma_{\phi}=0.924$, in which the user interference due the non-orthogonal transmission is less harmful than in the previous examples. 
The error rate graphs corresponding to the joint probability distributions $\underline{p}_{UV, \text{Case}1}$ and $\underline{p}_{UV, \text{Case}2}$ are plotted in Figs. \ref{2D090101} and \ref{2D040205}, respectively. 
From the simulation results, the conventional antipodal BPAM is found to be adequate because the ambiguity has been resolved. 
However, due to the non-uniform distribution of the sources, the individually optimized BPAM is still better than the conventional BPAM. 
Moreover, for both source distributions, our BPAM designs provide nearly optimal performance in the high SNR region. 
We note from Figs.\ \ref{2D090101} and \ref{2D040205} that at an error rate of $10^{-5}$, our joint designs achieve about a $1$dB and $2$dB SNR gain, respectively, over the individually optimized design.  
The minor performance degradation observed in the low SNR region is in fact the drawback of using the union bound in signal design. 

\subsection{Other Results}
In Fig.\ \ref{2DVarpulseEx1}, we depict the decoding performance under various correlation values between the senders' basic pulse waveforms for the joint probability distribution $\underline{p}_{UV, \text{Case}2}$. 
The values $\gamma_{\phi}=0, 0.383, 0.707, 0.924$ and $1$ correspond to the angles $\theta=\pi/2, 3\pi/8, \pi/4, \pi/8$ and $0$, respectively. 
As expected, the smaller $\gamma_{\phi}$, the better the decoding performance of the conventional BPAM and individually optimized BPAM. 
We observe that the orthogonal waveform transmission provides the best performance as intuitively expected. 
Moreover, it is worth mentioning that our design for $\gamma_{\phi}=1$ is better than both the conventional and individually optimized BPAM for $\gamma_{\phi}=0.924$ in the high SNR regime. This observation has practical importance because if orthogonal waveform transmission is unavailable for resource-limited networks, we may simply employ identical waveforms, i.e., $\gamma_{\phi}=1$, with our proposed constellations to achieve a good decoding performance. 
This way, the receiver only needs one matched filter for processing received signals. 

Lastly, we present an example in which the two senders transmit their signals with different average energy $E_1=2E_2$. 
Here, the case $\underline{p}_{UV, \text{Case}1}$ and $\gamma_{\phi}=1$ is considered, and the simulation results are depicted in Fig.\ \ref{1DUP}.
Clearly, due to the unequal energy allocation, the signal sets of both conventional antipodal BPAM and the individually optimized BPAM for the two senders are distinct, thereby yielding a better error rate performance than that of equal energy allocation (see Fig.~\ref{1D090101}).  
At an error rate of $10^{-5}$, the jointly optimized constellation achieves about $3$dB gain over the individually optimized design. 
Also, there is about $1$dB SNR gain from the unequal energy allocation for our design (compare Figs.~\ref{1D090101} and \ref{1DUP}). 
This example demonstrates that combining the energy allocation scheme \cite{JH13} with our design for two-sender GMAC can further improve the decoding performance. 

\section{Conclusions}
In this work, we investigated the design of optimized binary signaling schemes for sending correlated binary sources over non-orthogonal GMAC. 
For a wide range of SNRs and correlated source distributions, the error rate performance of the analytically derived signaling schemes was found to be quite close to the optimal performance under joint MAP decoding. 
In our experiments, the SNR gain achieved by our schemes is at least $2$dB over the individually optimized design for an error rate of $10^{-5}$ and strong interference.   
Future research directions include optimal energy allocation for different senders, nonbinary signaling, signaling design for the GMAC with more than two senders, fading channels, and constellation design for coded transmission. 

\appendix \label{AppendixA}
We consider the case $D_{10, 00}>0$, $D_{01, 00}>0$, and $|D_{10, 00}|-|D_{01, 00}|>0$, and evaluate (\ref{Pb1D}) explicitly. 
For other cases, the same procedure applies. 
Recall that $d_1=S_{11} - S_{10}$ and $d_2 = S_{21}-S_{20}$. 
Note that for the considered case we have $S_{10}<S_{11}$, $S_{20}<S_{21}$, and $S_{21}-S_{20}<S_{11}-S_{10}$. 
Based on these conditions, each of the four events in (\ref{MAPPb0}) is investigated as follows. 
Define $\mathcal{E}_1\triangleq\{\Delta_{00, 10}<0\}$, $\mathcal{E}_2\triangleq\{\Delta_{00, 01}<0\}$, and $\mathcal{E}_3\triangleq\{\Delta_{00, 11}<0\}$

\begin{itemize}[\IEEEsetlabelwidth{case\ }]
\item[Case 1.] When $(U, V)=(0, 0)$, the ranges of $N$ specified by $\mathcal{E}_1$, $\mathcal{E}_2$, and $\mathcal{E}_3$ are respectively given by
\begin{IEEEeqnarray}{lCl}
\text{Re}[N]&<& \frac{1}{d_1}\cdot\left[\sigma^2\ln\frac{p_{00}}{p_{10}} + \frac{d_1^2}{2}\right]\triangleq q_{11}\nonumber\\
\text{Re}[N]&<& \frac{1}{d_2}\cdot\left[\sigma^2\ln\frac{p_{00}}{p_{01}} + \frac{d_2^2}{2}\right]\triangleq q_{12}\nonumber\\
\text{Re}[N]&<& \frac{1}{(d_1+d_2)}\cdot\left[\sigma^2\ln\frac{p_{00}}{p_{11}} + \frac{(d_1+d_2)^2}{2}\right]\triangleq q_{13},\nonumber
\end{IEEEeqnarray}
Since $\text{Re}[N]$ is a zero mean Gaussian random variable with variance $\sigma^2$, we then have 
\begin{IEEEeqnarray}{rCl}
P_{\text{c}, 00}&=&\Pr(\text{Re}[N]<\min(q_{11}, q_{12}, q_{13}))\nonumber\\
&=&Q\left(\frac{-\min(q_{11}, q_{12}, q_{13})}{\sigma}\right).\nonumber
\end{IEEEeqnarray}

\item[Case 2.] When $(U, V)=(1, 0)$, the constraints on $N$ specified by $\mathcal{E}_1$, $\mathcal{E}_2$, and $\mathcal{E}_3$ are respectively given by
\begin{IEEEeqnarray}{rCl}
\text{Re}[N]&>& \frac{-1}{d_1}\cdot\left[\sigma^2\ln\frac{p_{10}}{p_{00}} + \frac{d_1^2}{2}\right]\triangleq q_{21}\nonumber\\
\text{Re}[N]&<& \frac{1}{d_2}\cdot\left[\sigma^2\ln\frac{p_{10}}{p_{11}} + \frac{d_2^2}{2}\right]\triangleq q_{22}\nonumber\\
\text{Re}[N]&>& \frac{1}{d_2-d_1}\cdot\left[\sigma^2\ln\frac{p_{10}}{p_{01}} + \frac{(d_2-d_1)^2}{2}\right]\triangleq q_{23}. \nonumber
\end{IEEEeqnarray}
According to the values of $q_{21}$, $q_{22}$, and $q_{23}$, we have
\[
P_{\text{c}, 10}=\left\{
\begin{array}{ll}
Q\left(\frac{q_{21}}{\sigma}\right)-Q\left(\frac{q_{22}}{\sigma}\right), &\text{if}\ q_{23}<q_{21}<q_{22}\\
Q\left(\frac{q_{23}}{\sigma}\right)-Q\left(\frac{q_{22}}{\sigma}\right), &\text{if}\ q_{21}<q_{23}<q_{22}\\
0, &\text{otherwise.}
\end{array}
\right.\]

\item[Case 3.] When $(U, V)=(0, 1)$, we obtain 
\begin{IEEEeqnarray}{rCl}
\text{Re}[N]&>& \frac{-1}{d_2}\cdot\left[\sigma^2\ln\frac{p_{01}}{p_{00}} + \frac{d_2^2}{2}\right]\triangleq q_{31}\nonumber\\
\text{Re}[N]&<& \frac{1}{d_1}\cdot\left[\sigma^2\ln\frac{p_{01}}{p_{11}} + \frac{d_1^2}{2}\right]\triangleq q_{32}\nonumber\\
\text{Re}[N]&<& \frac{1}{d_1-d_2}\cdot\left[\sigma^2\ln\frac{p_{01}}{p_{10}} + \frac{(d_1-d_2)^2}{2}\right]\triangleq q_{33}.\nonumber
\end{IEEEeqnarray}
Therefore, 
\[
P_{\text{c}, 01}=\left\{
\begin{array}{ll}
Q\left(\frac{q_{31}}{\sigma}\right)-Q\left(\frac{q_{33}}{\sigma}\right), &\text{if}\ q_{31}<q_{33}<q_{32}\\
Q\left(\frac{q_{31}}{\sigma}\right)-Q\left(\frac{q_{32}}{\sigma}\right), &\text{if}\ q_{31}<q_{32}<q_{33}\\
0, &\text{otherwise.}
\end{array}
\right.\]

\item[Case 4.] When $(U, V)=(1, 1)$, we have 
\begin{IEEEeqnarray}{rCl}
\text{Re}[N]&>& \frac{-1}{d_1}\cdot\left[\sigma^2\ln\frac{p_{11}}{p_{01}} + \frac{d_1^2}{2}\right]\triangleq q_{41}\nonumber\\
\text{Re}[N]&>& \frac{-1}{d_2}\cdot\left[\sigma^2\ln\frac{p_{11}}{p_{10}} + \frac{d_2^2}{2}\right]\triangleq q_{42}\nonumber\\
\text{Re}[N]&>& \frac{-1}{d_1+d_2}\cdot\left[\sigma^2\ln\frac{p_{11}}{p_{00}} + \frac{(d_1+d_2)^2}{2}\right]\triangleq q_{43}.\nonumber
\end{IEEEeqnarray}
Consequently, we have
\begin{IEEEeqnarray}{rCl}
P_{\text{c}, 11}&=&\Pr(\text{Re}[N]>\max(q_{41}, q_{42}, q_{43}))\nonumber\\
&=&Q\left(\frac{\max(q_{41}, q_{42}, q_{43})}{\sigma}\right).\nonumber
\end{IEEEeqnarray}
\end{itemize}

Using the above results, the decoding performance for one-dimensional combined constellation can be obtained via (\ref{MAPPb0}).

\begin{figure}[!tb]
\centering
\includegraphics[draft=false, scale=0.6]{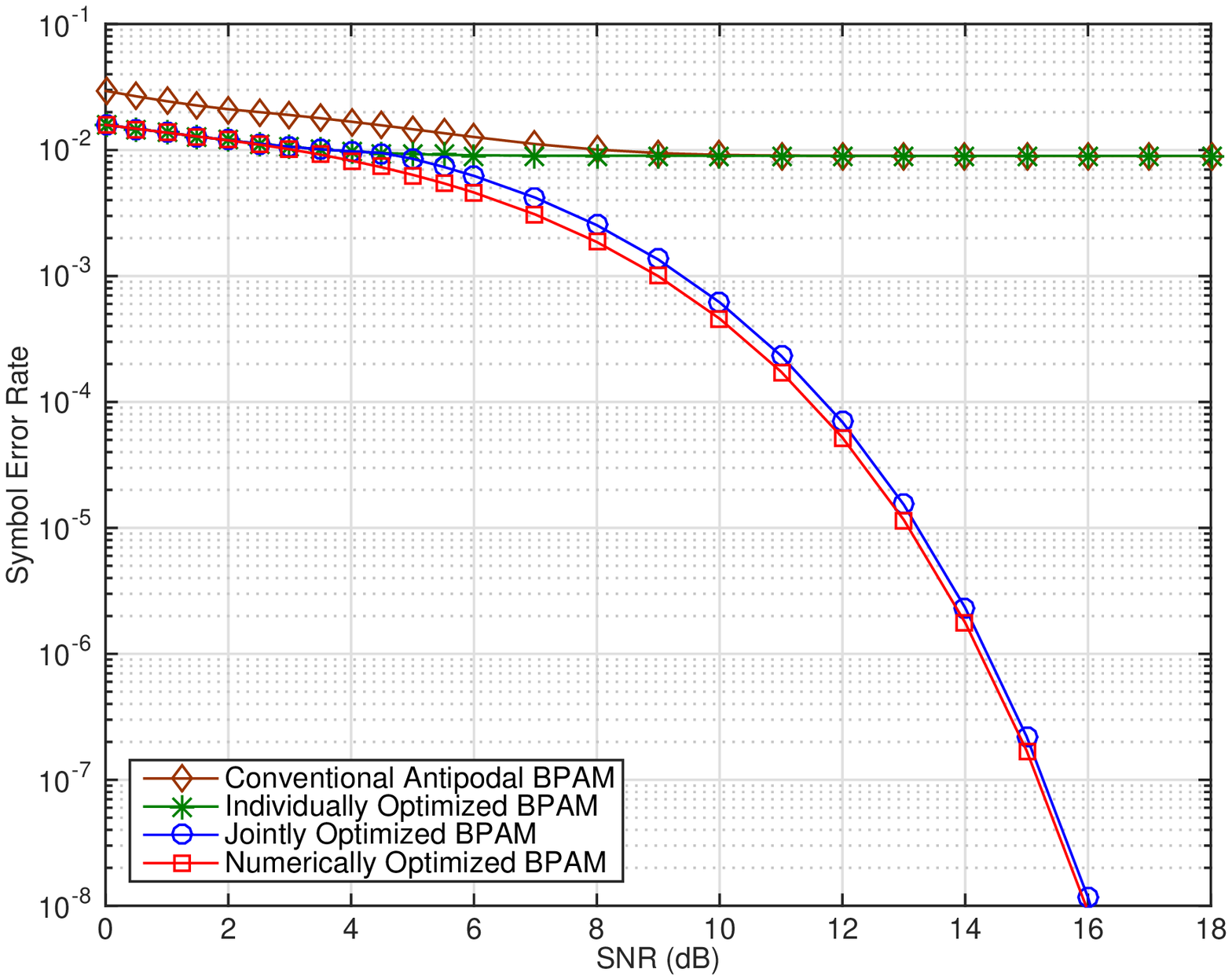} 
\caption{The decoding performance of various BPAM for correlated sources with $\underline{p}_{UV, \text{Case}1}$, where $\gamma_{\phi}=1$ and $E_1=E_2$.}
\label{1D090101}
\end{figure}

\begin{figure}[!thb]
\centering
\includegraphics[draft=false, scale=0.6]{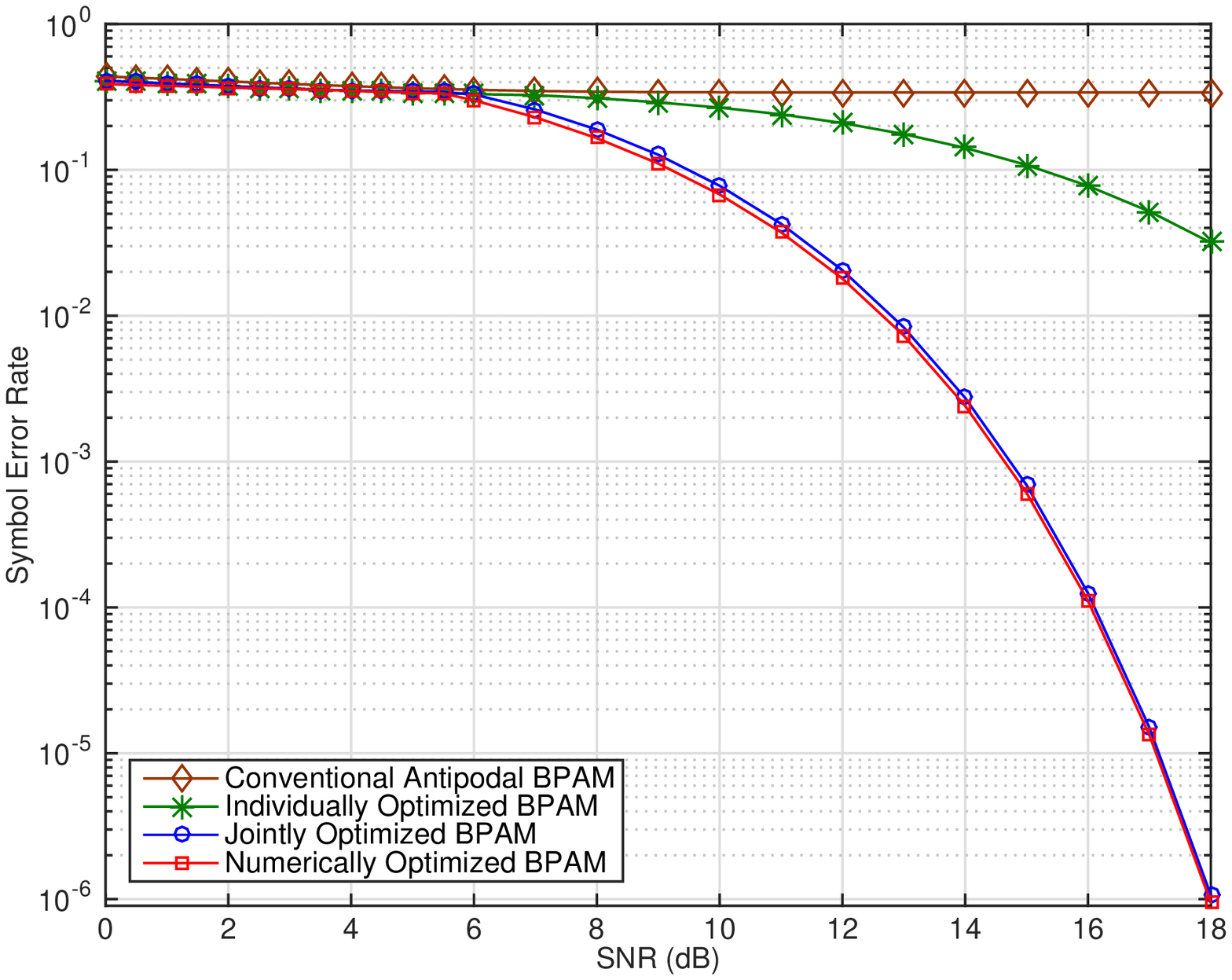} 
\caption{The decoding performance of various BPAM for correlated sources with $\underline{p}_{UV, \text{Case}2}$, where $\gamma_{\phi}=1$ and $E_1=E_2$.}
\label{1D040205}
\end{figure}

\begin{figure}[!tb]
\centering
\includegraphics[draft=false, scale=0.6]{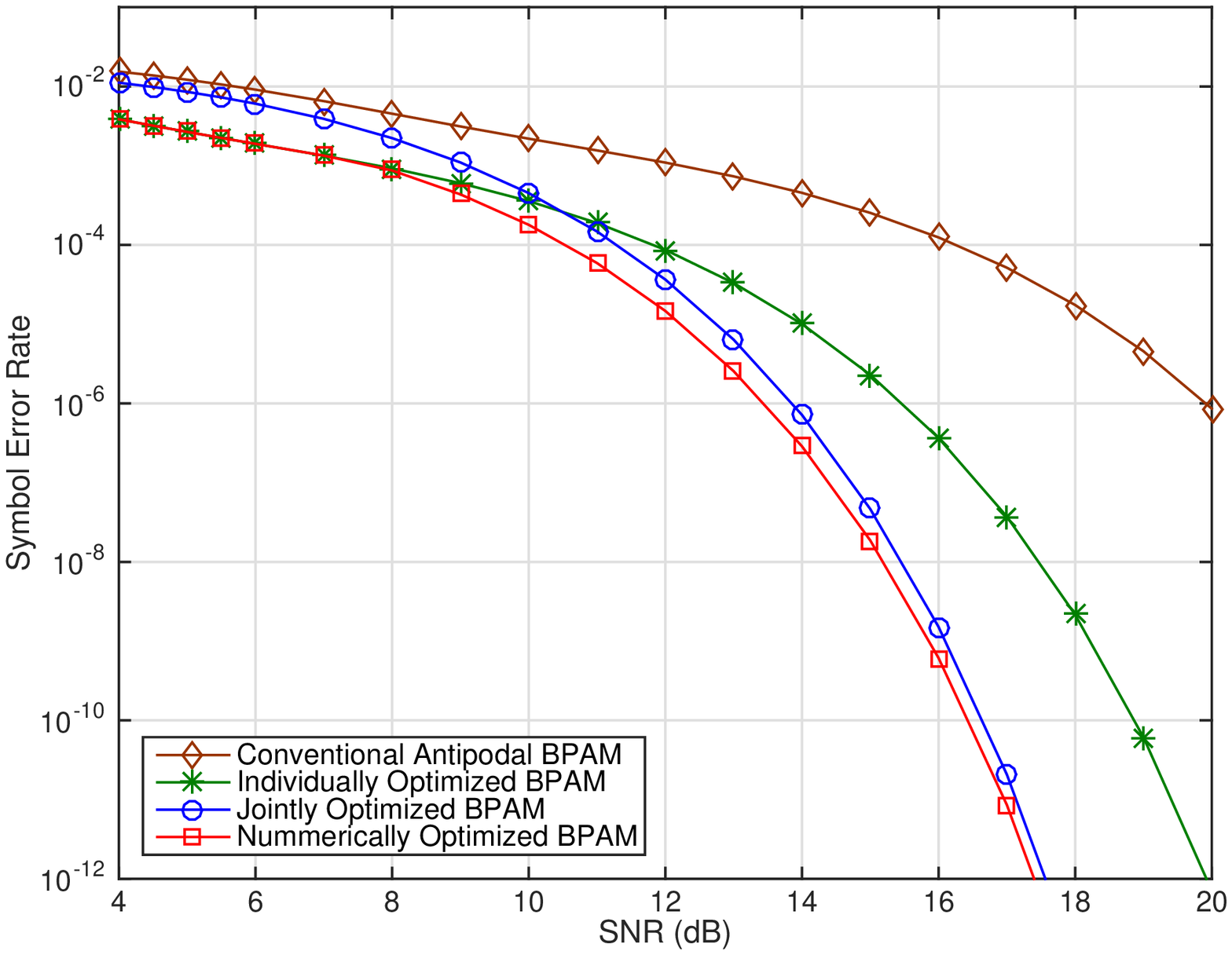} 
\caption{The decoding performance of various BPAM for $\underline{p}_{UV, \text{Case}1}$, where $\gamma_{\phi}=0.924$ and $E_1=E_2$.}
\label{2D090101}
\end{figure}

\begin{figure}[!thb]
\centering
\includegraphics[draft=false, scale=0.6]{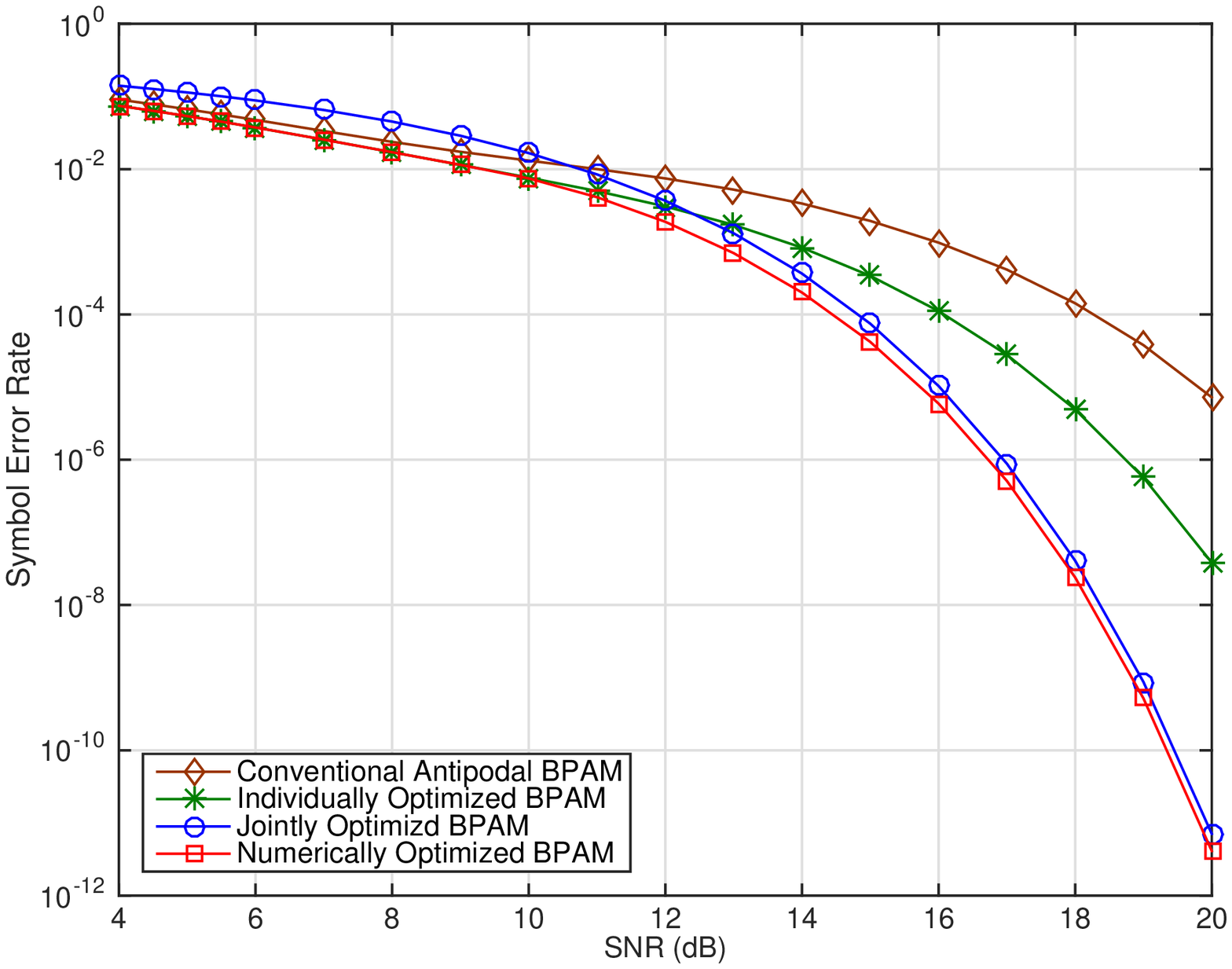} 
\caption{The decoding performance of various BPAM for $\underline{p}_{UV, \text{Case}2}$, where $\gamma_{\phi}=0.924$ and $E_1=E_2$.}
\label{2D040205}
\end{figure}

\begin{figure}[!tb]
\centering
\includegraphics[draft=false, scale=0.6]{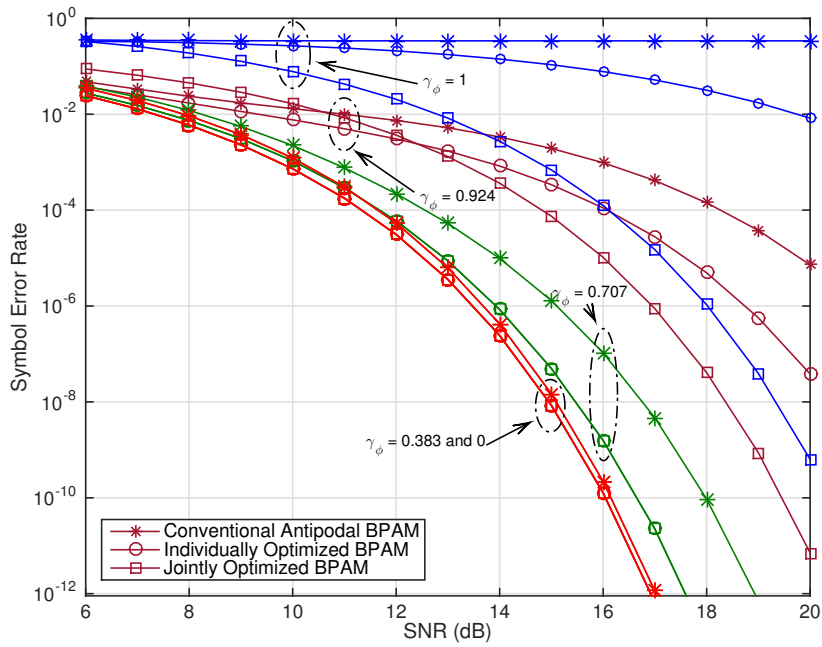} 
\caption{The decoding performance of various BPAM and various $\gamma_{\phi}$ for $\underline{p}_{UV, \text{Case}2}$, where $E_1=E_2$.}
\label{2DVarpulseEx1}
\end{figure}

\begin{figure}[!thb]
\centering
\includegraphics[draft=false, scale=0.6]{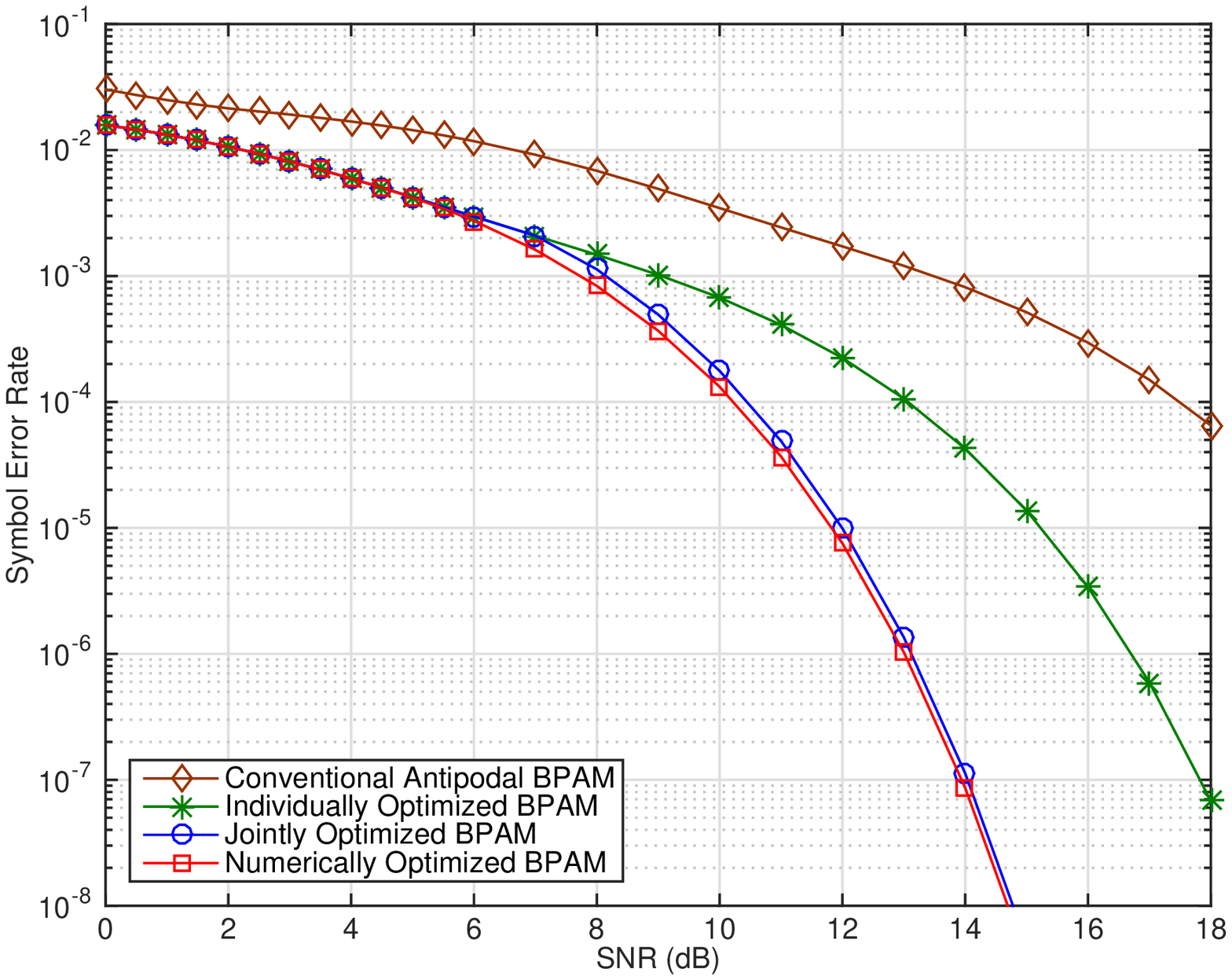} 
\caption{The decoding performance of various BPAM for $\underline{p}_{UV, \text{Case}1}$, where $\gamma_{\phi}=1$ and $E_1=2E_2$.}
\label{1DUP}
\end{figure}

\end{document}